\documentclass[a4paper, journal]{IEEEtran}
\IEEEoverridecommandlockouts
\usepackage{cite}
\usepackage{amsmath,amssymb,amsfonts}
\usepackage{algorithmic}
\usepackage{graphicx}
\usepackage{textcomp}
\usepackage{xcolor}
\usepackage{float}
\usepackage{paralist}
\usepackage{amsthm}
\usepackage{array}
\usepackage{booktabs, multirow} %
\usepackage{soul} %
\usepackage{xcolor,colortbl} %
\usepackage{changepage,threeparttable} %
\usepackage{longtable}
\usepackage{amsmath,amsfonts}
\usepackage{algorithmic}
\usepackage{algorithm}
\usepackage{array}
\usepackage{tikz}

\newtheorem{theorem}{Theorem}
\newtheorem{proposition}{Proposition}
\def\BibTeX{{\rm B\kern-.05em{\sc i\kern-.025em b}\kern-.08em
    T\kern-.1667em\lower.7ex\hbox{E}\kern-.125emX}}
\begin{document}
\title{Inverse Learning assisted V2I Communication  for Intent Based 6G ISAC Vehicular Networks\\
\thanks{A preliminary version of this work was presented at IEEE WCNC 2026~\cite{CV2604:Inverse}.}
}

\author{\IEEEauthorblockN{ Anoop C V and Anup Aprem\\}
\IEEEauthorblockA{\textit{Dept. of Electronics and Communication Engineering},
\textit{National Institute of Technology Calicut},
India  \\
anoop\_p210065ec@nitc.ac.in, anup.aprem@nitc.ac.in}
    
}

\maketitle
\begin{abstract}
6G communication is expected to bring unprecedented advancements in the capabilities and efficiency of vehicular networks. However, the advent of 6G will also introduce significant changes in the operation of vehicular communication infrastructures such as roadside units (RSUs), including the incorporation of autonomous intent-based network paradigm and integrated sensing and communication (ISAC) capabilities. While ISAC enables sensing and communication capabilities within a single 6G network node, intent-based network design paradigm ensures that network nodes such as RSUs, act as autonomous cognitive agents to fulfill the objectives of their respective communication service providers. This paradigm shift necessitates the development of vehicle to infrastructure (V2I) communication strategies that learns and adapts to the sensing-assisted communication and the autonomous decision-making strategies of RSUs. 
We model the RSU as a constrained utility maximizer, where the utility function characterizes the RSU intent, and formulate an inverse learning (IL) problem to infer the underlying utility function from observed (revealed) ISAC RSU actions, for example the adaptive beamwidth allocation in response to the kinematic states of vehicles within a vehicular micro-cloud (VMC).
The main contributions of this paper are: (i) ATIL, a nonparametric method based on Afriat’s theorem for fixed utility learning; (ii) FICNNIL, a parametric approach using fully input-concave neural networks, for structured fixed utility learning; and (iii) PICNNIL, a parametric approach based on partially input-concave neural networks, for inverse learning of state-dependent utilities. (iv) Federated inverse learning algorithms FedFICNNIL and FedPICNNIL for fixed and state dependent utility, respectively.
We demonstrate the proposed IL-based framework for two V2I communication applications in VMCs, namely predictive scheduling for cooperative data downloading and dynamic cluster-head selection, under both fixed and state dependent RSU utility models. Numerical results show that the proposed IL approaches consistently outperform non-learning baselines in terms of throughput and packet loss. Moreover, ATIL, FICNNIL, PICNNIL, FedFICNNIL and FedPICNNIL outperform representative parametric and classical supervised learning methods for inverse learning. 

\end{abstract}

\begin{IEEEkeywords}
6G Vehicular Network, Vehicular Micro-Clouds, Adaptive communication, Predictive communication strategies, Integrated Sensing and Communication (ISAC), Machine learning/Network Intelligence for ISAC, ISAC for vehicular-to-everything (V2X) networks, Utility Maximization, Inverse Learning.
\end{IEEEkeywords}

\section{Introduction}

The upcoming sixth generation (6G) of mobile communication networks, envisioned under the International Telecommunication Union (ITU) framework for IMT-2030 \cite{ITU_R_June_2023}, is expected to support novel usage scenarios such as ubiquitous connectivity, intent based networks, and Integrated Sensing and Communication (ISAC).  
In intent-based network design \cite{ericson2023IntentNW,zhang2023intent6G,tmforum_intent}, network nodes act as autonomous cognitive agents that align their operations with the intent (or objectives) of communication service providers (CSPs).
While, ISAC enables sensing and communication capabilities within a single 6G network node \cite{ETSI_GR_ISC}.
Several recent research efforts have demonstrated the effectiveness of ISAC-enabled, reconfigurable Roadside Units (RSUs) in enabling adaptive, sensing-assisted vehicular communication \cite{Liu2023ISAC_VN_Chapter,ISACVNC2022funda}. 

Our work is motivated by futuristic vehicular communication use cases such as vehicular platoons and vehicular micro-clouds (VMC) for improved communication efficiency and the realization of safe and economic transport, initially recommended by the 5G Automotive Association (5GAA) and subsequently extended by ETSI through the integration of ISAC \cite{5GAA_ETSI,ETSI_GR_ISC,Platoon2024proximity3}.\footnote{For further details, the interested reader is referred to the technical reports of ITU-R Working Party 5D on IMT-2030 \cite{ITU_R_June_2023}, ETSI Industry Specification Group (ISG) on ISAC \cite{ETSI_GR_ISC}, and vehicular communication \cite{5GAA_ETSI}.} 
For instance, cooperative data downloading based on VMCs has been shown to achieve up to $70\%$ improvement in median download time, energy efficiency per unit data transmitted, and packet loss ratio, while ensuring reliable delivery within the vehicle--RSU contact duration \cite{survey2019coopDownlAdvant}.
Below, we first describe the challenges for vehicular communication in an intent based ISAC RSU, followed by our contributions.

{\sl \textbf{Intent based ISAC RSU}: }
At any given time instance, an intent based ISAC RSU senses (e.g., kinematic states of the vehicles) its dynamic environment and takes decisions by adaptively switching waveform parameters or dynamically allocating resources
The \emph{rational behavior} of ISAC RSU can be mathematically modeled as a constrained optimization problem \cite{ericson2023IntentNW,zhang2023intent6G,tmforum_intent} of the form
\begin{equation}\label{eq:RSU_Strategy_general}
    \max_{\boldsymbol{\beta} \in \Pi_{\boldsymbol{\alpha}}} 
    \; U(\boldsymbol{\cdot}),
\end{equation}
where $\boldsymbol{\beta}$ denotes the decision variables (e.g., waveform parameters, resource allocation variables), $U(\boldsymbol{\cdot})$ denotes the \emph{utility function}, and $\Pi_{\boldsymbol{\alpha}}$ denotes the feasible set, with the subscript ${\boldsymbol{\alpha}}$ parameterizing the state of the environment (e.g., position, velocity, or their uncertainty of the vehicles, sensed by the RSU). 
The utility function $U(\boldsymbol{\cdot})$ characterizes the intent of the RSU, and in practice can be a function of the decision variable \(\boldsymbol{\beta}\) alone (referred to as \emph{fixed utility}, \(U(\boldsymbol{\beta})\), in this paper) or both \(\boldsymbol{\beta}\) and the state of the world \(\boldsymbol{\alpha}\) (referred to as \emph{state dependent utility, \(U(\boldsymbol{\alpha},\boldsymbol{\beta})\)} in this paper).
$U(\boldsymbol{\cdot})$ in \eqref{eq:RSU_Strategy_general} is typically \emph{concave} and (componentwise) \emph{monotonically increasing} in $\boldsymbol{\beta}$ so as to reflect saturation effects in performance gains (i.e., additional resource/waveform refinement yields progressively smaller improvements) while ensuring that better QoS/KPI performance is always preferred. 
Further, the feasible set $\Pi_{\boldsymbol{\alpha}}$ captures the underlying resource limitations or constraints imposed by QoS requirements or KPI of the service, which govern the RSU's adaptive behavior given the environmental state $\boldsymbol{\alpha}$.

{\sl \textbf{Challenges in V2I communication in Intent based ISAC RSU: }}
In intent-based vehicular networks comprising autonomous ISAC RSUs with adaptive decision-making and sensing capabilities, and capable of strategic behavior, the design of adaptive vehicle to infrastructure (V2I) communication strategies must explicitly account for the RSU's decision-making mechanism. However, the RSU decision-making strategy, characterized by the utility function in \eqref{eq:RSU_Strategy_general}, may not be explicitly known to the vehicles. \emph{Consequently, vehicles must employ data-driven approaches to infer the RSU strategy (utility) by observing revealed RSU behavior---such as waveform parameters, beam time, and communication data rate---denoted by $\boldsymbol{\beta}$, selected by the ISAC RSU in response to the state of the vehicles $\boldsymbol{\alpha}$, in a non-cooperative setting. We refer to this problem as ``Inverse Learning (IL)''.}

While developing a data driven IL framework for vehicular communication network contexts, an important aspect is whether the dataset for training the IL model is available at a centralized server or distributed in a decentralized manner across different VMCs (referred to as clients), which may or may not be willing to share the data due to communication overhead or privacy concerns---leading to centralized or decentralized (federated) learning paradigms.

In this context, the key contributions of this paper are summarized as follows:
\begin{enumerate}[I.]
\item We develop three inverse learning frameworks:
\begin{enumerate}[(i)]
    \item \textbf{ATIL (Afriat's-Theorem-based Inverse Learning):} a \emph{nonparametric} approach that reconstructs a utility function consistent with the observed ISAC RSU behaviour using Afriat's theorem from revealed preference theory; suitable for \emph{fixed RSU utilities} under centralized learning paradigms.
    \item \textbf{FICNNIL (FICNN-based Inverse Learning):} a \emph{parametric} approach that models the RSU utility using a Fully Input Concave Neural Network (FICNN) \cite{amos2017input_ICNN_Seminal,grzeskiewicz2025uncovering_ICNN_Utility_Max} and leverages revealed preference (see Sec.~\ref{subsec:ATIL}) based constraints for training; suitable for \emph{fixed RSU utilities} in both \emph{centralized and decentralized learning paradigms}.
    \item \textbf{PICNNIL (PICNN-based Inverse Learning):} a \emph{parametric} approach that models the RSU utility using a Partially Input Concave Neural Network (PICNN) \cite{amos2017input_ICNN_Seminal} to capture \emph{state-dependent utilities}, and leverages revealed preference based constraints for training; suitable for both \emph{centralized and decentralized learning paradigms}.
    \item We propose a novel bi-level training process for training FICNN and PICNN - see Sec.~\ref{subsec:ICNNTraining} for details.
\end{enumerate}

\item We propose FedFICNNIL and FedPICNNIL, algorithms for federated inverse learning of ISAC RSU nodes in a decentralized setting, where training data are not available at a centralized server and the available training observations are heterogeneous due to the state-dependent nature of the RSU utility.

\item We demonstrate ATIL, FICNNIL, and PICNNIL for two inverse-learning-assisted V2I communication applications in 6G-ISAC vehicular networks:
\begin{inparaenum}[(i)]
    \item predictive scheduling for cooperative data downloading in VMCs, and
    \item dynamic cluster-head selection in VMCs.
\end{inparaenum}
    
\item Numerical results show that the proposed IL-based approaches (ATIL, FICNNIL, PICNNIL, FedFICNNIL and FedPICNNIL) outperform 
existing techniques in the literature for predictive scheduling and cluster head selection
in terms of throughput and packet loss for both applications. Further, ATIL, FICNNIL, PICNNIL, FedFICNNIL and FedPICNNIL also outperform representative parametric baselines and standard supervised machine-learning approaches.
\end{enumerate}

\emph{To the best of our knowledge, no existing literature addresses the problem of developing data-driven V2I communication strategies within a 6G-ISAC intent-based autonomous network with unknown RSU strategy.}

The rest of this paper is organized as follows: Section~\ref{sec:Background} outlines the related literature that motivates the formulation of our problem.
Sec.~\ref{sec:SysModel} provides an overview of the 6G-ISAC vehicular network infrastructure considered in this paper. Sec.~\ref{sec:ProposedMethod} presents the proposed inverse learning methodologies in both centralized and federated learning paradigms, and also presents the IL assisted adaptive V2I strategies. Subsequently, Sec.~\ref{sec:NumRes} demonstrates the application of IL-based strategies for adaptive selection of the VMC cluster head and predictive scheduling in cooperative data downloading within the VMC, and compares the results with state-of-the-art machine learning approaches, and existing non-learning based communication strategies. Finally, Sec.~\ref{sec:Conclusion} concludes the paper.

\section{Related Literature}\label{sec:Background}
{\sl \textbf{ISAC for vehicular networks}: }6G vehicular communication should support very high data rate services such as real time video transmission and HD map updates -- demanding very stringent QoS requirements namely ultra high throughput (at least $100$ Gbps), high reliability (at least \(99.9999\%\)) ultra low delay (less than \(0.1\)ms) \cite{Liu2025ISACVNCSurveyCommincationOpenCall}. One among the key enablers of these QoS requirements is precise massive MIMO (mMIMO) beam forming with sharp pencil-beam pattern and low side lobes, which in turn requires precise target localization\cite{Liu2025ISACVNCSurveyCommincationOpenCall,Liu2023ISAC_VN_Chapter,ISACVNC2022funda}. ISAC enabled beamforming can solve the inherent spectral efficiency-localization trad eoff of the conventional pilot based communication only localization protocols\cite{ISACVNC2022funda,Liu2025ISACVNCSurveyCommincationOpenCall}.

In high-mobility vehicle-to-infrastructure (V2I) scenarios, conventional pilot-based beam tracking can incur significant signaling overhead and may become unreliable due to rapid vehicle motion, narrow mmWave/mMIMO beams, and changing roadway geometry. ISAC-enabled roadside units (RSUs) address this challenge by reusing the communication waveform for sensing and processing the reflected echoes to estimate vehicle-related parameters such as position, velocity, and angular direction, which can then be used for predictive beam alignment and reliable infrastructure-to-vehicle beamforming~\cite{Liu2025ISACVNCSurveyCommincationOpenCall,liu2020radar,Liu2023ISAC_VN_Chapter,wang2021greenIOV,liu2023energy}. Recent studies have further extended ISAC-based vehicular systems toward extended-target tracking, roadway-geometry-aware beam management, multi-vehicle tracking, behavior-aware beamforming, hybrid beamforming, and learning-based joint sensing-communication optimization~\cite{Cong2023VehicularBehaviorISAC,Yu2023HybridBeamformingIoV,Liu2023ISAC_VN_Chapter}. These works show that the communication performance of an ISAC-enabled vehicular network is closely coupled with sensing and tracking accuracy, since inaccurate tracking can cause beam misalignment and throughput degradation, whereas excessive sensing may reduce the resources available for communication. Hence, adaptive sensing-communication resource allocation becomes essential for maintaining reliable vehicular connectivity. However, most existing works focus on the forward design of known RSU objectives or policies, while the inverse problem of learning the latent decision-making behaviour of an adaptive ISAC-enabled RSU from observed state-action data remains unexplored.

{\sl \textbf{Intent based reconfigurable ISAC RSU design}: }Intent Driven networks (IDN) focuses on conceiving elements of the 6G network as self-contained rational agents that maximize specific utility functions according to the strategies of the respective CSPs \cite{zhang2023intent6G,ericson2023IntentNW}, and can adaptively reconfigure communication and sensing parameters based on dynamic requirements. In the future 6G ecosystem, the utility of an autonomous network node will depend on the monetization plan of the CSP\cite{zhang2023intent6G}, as well as the relevance of achieving a certain level of compliance with the QoS requirements specified by IMT-2030 capabilities. Some of the key metrics of interest that different CSPs try to optimize during ISAC resource allocation \cite{al2024ISACresources,Liu2025ISACVNCSurveyCommincationOpenCall} are \begin{inparaenum}[(i)]
    \item probability of detection,
    \item error bounds of sensed estimates,
    \item beam misalignment, and
    \item overall data rate.
\end{inparaenum}
For example, an RSU equipped with ISAC capabilities and operating under a CSP that prioritizes sensing accuracy over high data rates will adopt a fundamentally different radio resource allocation and multiple access strategy compared to an RSU managed by a communication-centric CSP \cite{lu2024ISAC_challenges}.

{\sl \textbf{VMCs and platoons in intent based 6G-ISAC RSU Infrastructure}: } 
VMCs and vehicular platoons are intelligent connected vehicles, cooperating with each other as a cluster, for improved communication and traffic efficiency. 
Cooperative data downloading in VMC \cite{survey2019coopDownlAdvant} can enhance timely delivery of large data (e.g., media files, high-definition maps), thereby reducing packet loss and minimizing the need for frequent retransmission. \cite{chen2018timely,electronics11223663PreCache} show that predictive scheduling, which accounts for channel capacity variations, can enhance throughput and reduce packet loss in cooperative data downloading within VMCs.

However, the characteristics of the I2V channel from 6G ISAC RSU, such as capacity and reliability, depend on both the tracking accuracy (required for narrow beamforming) and the transmission strategy adopted by the RSU. Consequently, if a vehicle, via the vehicle-to-infrastructure (V2I) uplink, requests data exceeding the I2V downlink capacity of the RSU, this may lead to packet retransmission or data loss, thereby degrading the streaming quality and suboptimal power efficiency\cite{survey2019coopDownlAdvant}.
\emph{But the strategy adopted by the RSU may not be publicly known and vary from time to time based on the business needs of CSP. This paradigm shift introduces significant challenges in the design of communication protocols for future vehicular networks, necessitating the development of novel and adaptive V2I protocols for 6G ISAC scenarios, particularly when the I2V beamforming and multiple access strategies of RSUs are not explicitly known.}

Existing algorithms in the intelligent connected vehicle literature rely mainly on pilot based channel estimation techniques and exhibit limitations such as sub-optimal throughput, high packet loss, and increased retransmission \cite{hagenauer2019explainVMC}.
The existing works on cluster head selection and predictive scheduling focuses on non-intent based 5G RSUs and generally assume that the RSU strategy is known to the VMC. In this paper, we extend the \emph{cluster head selection} and \emph{predictive scheduling} problems in intelligent connected vehicles literature to intent based 6G-ISAC context. Furthermore, conventional supervised regression/classification models, which require a large amount of training data to learn RSU behaviour, cannot be directly employed in the IL context, since the relationship between the state of the world and the RSU is not governed by a simple functional mapping, but rather through a constrained optimization problem as given in~\eqref{eq:RSU_Strategy_general} \cite{anoop2025bayesian}.

\section{System Model}\label{sec:SysModel}
We consider a scenario in which ISAC-RSUs\footnote{For a detailed treatment of ISAC waveform design and the integration of sensing and communication using a common waveform, see \cite{Liu2023ISAC_VN_Chapter,ISACVNC2022funda}.} equipped with mMIMO and beam steering capability serves VMC clusters comprising \( M \) different vehicles within their coverage area. The sensing aided vehicular communication strategy adopted by a single RSU is illustrated in Fig.~\ref{RSU_V2I}.  
In the following, we first mathematically model the 6G Vehicular network infrastructure in Sec.~\ref{subsec:VehNetInfr}. Further, ISAC target tracking and sensing assisted adaptive beam forming are discussed in Sections~\ref{sec:ISACTracking} and \ref{sec:ISACAB} respectively.
\subsection{6G Vehicular Network Infrastructure}\label{subsec:VehNetInfr}
\begin{figure}[ht]
    \centering
    \includegraphics[width=0.8\linewidth,trim={0 0 0 0},clip]{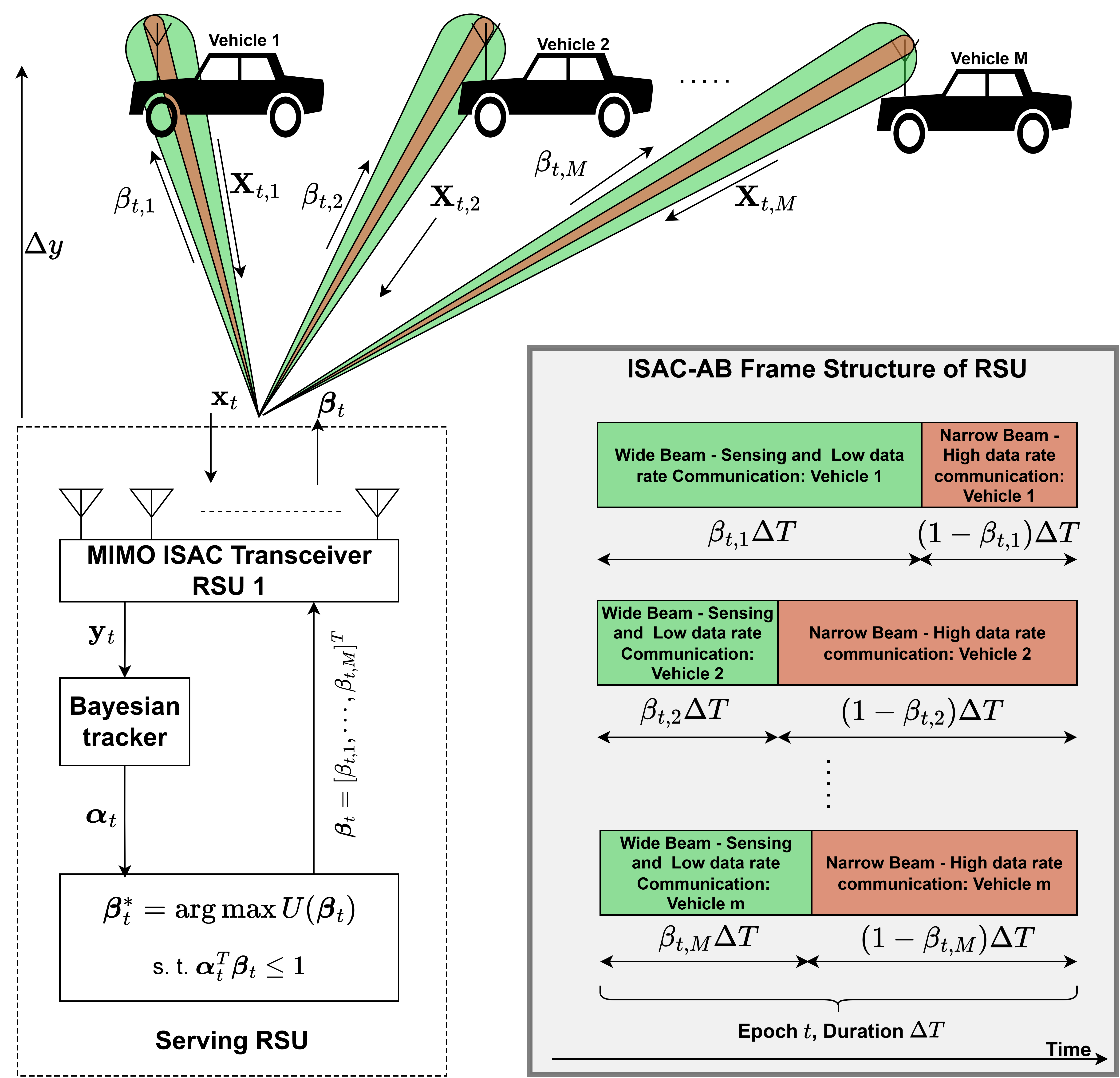}
    \caption{Sensing-aided 6G vehicular communication network with reconfigurable ISAC-RSU. The RSU, modeled as constrained utility maximizer (as in \eqref{eq:ISAC_AB_Objective}), employs ISAC adaptive beamforming strategy~\cite{Liu2023ISAC_VN_Chapter} to reconfigure the ISAC beamwidth according to the sensing accuracy of its Bayesian tracker. The beam allocation strategy of the RSU, and hence the maximum capacity of the I2V link, is unknown to the vehicles. Thus, any adaptive V2I communication strategies must involve \emph{inverse learning} of the RSU strategy from past observations.}
    \label{RSU_V2I}\vspace{-0.5em}
\end{figure}
We assume that the RSUs adopt ISAC Adaptive Beamforming (ISAC AB), allowing joint vehicle tracking and vehicular communication \cite{Liu2023ISAC_VN_Chapter} in a time-division manner. As shown in Fig.~\ref{RSU_V2I}, in a given epoch \( t \), an RSU, utilizes a Bayesian tracker to estimate the angle \( (\phi_{t,m}) \), distance \( (d_{t,m}) \), and velocity \( (v_{t,m}) \) of the vehicle \(m \in \{1,2,\cdots,M\}\), (collectively represented as \(\mathbf{X}_{t,m}\) in Fig.~\ref{RSU_V2I}) in a VMC, based on the echoes of wide beam ISAC signals reflected from it.
ISAC AB resource allocation strategy in RSU, as depicted by the frame structure in Fig.~\ref{RSU_V2I}, splits each epoch into two time slots, wherein the first slot \( k \in [(t-1)\Delta T, (t-1+\beta_{t,m})\Delta T] \) transmits a wide ISAC beam for accurate sensing and low data rate communication, and the second slot \( k \in ((t-1+\beta_{t,m})\Delta T, t\Delta T] \) transmits a high data rate narrow beam solely for communication, where \( \beta_{t,m} \in [0,1] \) is the splitting factor for vehicle \(m\), and \(\Delta T\) is duration of an epoch.
For example, a sensing-centric CSP may aim to keep \( \beta_{t,m} \) close to $1$ to maximize sensing accuracy, while a communication-centric CSP may optimally select a splitting ratio that minimizes beam misalignment loss while maximizing overall system throughput. Hence, different CSPs may adopt strategies aligned with their intentions to determine the optimal splitting ratios.

\subsection{ISAC Target Tracking}\label{sec:ISACTracking}
The RSU is located at origin, while the \( M \) targets, each modeled as point scatterers \cite{Liu2023ISAC_VN_Chapter}, move within the footprint \(((x_{\max},\Delta y) \text{ to } (x_{\min},\Delta y))\) of the RSU at a nearly constant velocity. RSU use Bayesian filters\cite{bayesianFilters2003survey} to continuously track the state of the target \( m \) at any epoch \( t \),  (comprising angle, distance, and velocity), \( \mathbf{X}_{t,m} = [\phi_{t,m}, d_{t,m}, v_{t,m}]^T \). The state evolution model and the linear Gaussian measurement model are given in \eqref{eq:NLTarDynMod} and \eqref{eq:LinGausObsMod}, respectively:
\begin{align}\label{eq:NLTarDynMod}
    \mathbf{X}_{t,m} &= \mathbf{h}\left(\mathbf{X}_{t-1,m}\right) + \mathbf{w}_{t,m}\text{; }\mathbf{w}_{t,m}\sim \mathcal{N}(0,Q),\\
    \mathbf{y}_{t,m} &= \mathbf{X}_{t,m} + \nu_{t,m}\text{; }\mathbf{\nu}_{t,m}\sim \mathcal{N}(0,R),\label{eq:LinGausObsMod}
\end{align}
where \( Q\in\mathbb{R}^{3\times 3} \) and \( R\in\mathbb{R}^{3\times 3} \) are standard process noise covariance and observation noise covariance matrices, respectively \cite{Liu2023ISAC_VN_Chapter}. Based on the observations \( \mathbf{y}_{1:t,m} \), the RSU estimates its belief of the state of target \( m \):
\begin{equation}\label{eq:EKFBelief}
    \pi_{t,m} = P\left({\mathbf{X}}_{t,m}\mid\mathbf{y}_{1:t,m}\right).
\end{equation}
If the state evolution model $\mathbf{h}(\cdot)$ is linear, then the RSU employs a standard Kalman filter, and if it is non-linear, the RSU uses extended Kalman Filter (EKF) \cite{Liu2023ISAC_VN_Chapter} to approximate \eqref{eq:EKFBelief} as \( \pi_{t,m} \sim \mathcal{N}(\hat{\mathbf{X}}_{t,m},\boldsymbol{\Sigma}_{t,m}) \) by linearizing \( \mathbf{h(\cdot)} \), where \( \hat{\mathbf{X}}_{t,m}\in \mathbb{R}^3 \) is the conditional mean state estimate, and \( \boldsymbol{\Sigma}_{t,m}\in \mathbb{R}^{3\times 3} \) is the covariance of the state of target \( m \) at any given epoch \( t \).

\subsection{Adaptive Beam Allocation in an intent based ISAC RSU}\label{sec:ISACAB} 
Let \( \boldsymbol{\beta}_t = [\beta_{t,1},\beta_{t,2},\cdots,\beta_{t,M}]^T \), denote the vector of splitting ratios (see Fig.~\ref{RSU_V2I} in Sec~\ref{subsec:VehNetInfr}) at each time for the RSU.
The RSU under the constraints imposed by it's CSP assigns $\beta_{t,m}$ by solving a utility maximization problem as in \eqref{eq:RSU_Strategy_general}. The utility function $U(\cdot)$ encodes the intent of the CSP operating the RSU. 
We abstractly model the utility maximization problem as follows. For each vehicle \(m\) the RSU employs, a Bayesian tracker to estimate the state $\mathbf{X}_{t,m}$ and the corresponding covariance estimate, $\boldsymbol{\Sigma}_{t,m}$, using \eqref{eq:EKFBelief}. The predictive covariance, $\boldsymbol{\Sigma}_{t|t-1,m}$, obtained using a standard Bayesian filtering recursion gives the measure of the uncertainty anticipated by the RSU in the estimate for the next epoch, and the accuracy of the target state estimate in epoch \(t\) is denoted as:
\begin{equation}\label{eq:pobeElem}
    \alpha_{t,m} = \text{trace}(\boldsymbol{\Sigma}_{t|t-1,m}^{-1}), \quad m = 1,2,\cdots,M.
\end{equation}
In each epoch, the RSU determines the optimal splitting ratio, \(\boldsymbol{\beta}_t^*\), by solving a constrained optimization problem as in \eqref{eq:RSU_Strategy_general} where the utility characterizes the RSU intent and the state of the environment parameterized by the vector of tracking accuracy \(\boldsymbol{\alpha}_t  = [\alpha_{t,1},\alpha_{t,2},\cdots,\alpha_{t,M}]^T\) will decide the constraints.

{\sl Fixed and state dependent RSU utility: }In this paper, we consider two types of RSU utilities, namely: \begin{inparaenum}[(i)]
    \item fixed RSU utility, \(U(\boldsymbol{\beta}_t)\): the utility depends only on the RSU decision variable \(\boldsymbol{\beta}_t\), and is independent of the state of the world parameterized by \(\boldsymbol{\alpha}_t,\) and
    \item state dependent RSU utility, \(U(\boldsymbol{\alpha}_t,\boldsymbol{\beta}_t)\): utility varies as a function of the state of the world \(\boldsymbol{\alpha}_t.\)
\end{inparaenum} In this paper, the state-dependent utility of an ISAC RSU captures the dependence of its preference between sensing and communication on the state of the world, parameterized by the overall vehicle tracking accuracy. For example, the tracking accuracy of an RSU serving VMCs on a straight highway with lanes having specific speed limits will be significantly higher than that of an RSU serving VMCs on roads with junctions, bends, and permission to overtake. State-dependent utility allows the RSU to assign very low priority to tracking when the overall tracking accuracy is good, and to allocate very high priority for accurate tracking when the overall tracking accuracy is low. In contrast, under a fixed RSU utility, the preference between sensing and communication remains fixed irrespective of the state \(\boldsymbol{\alpha}_t\). 

The constrained optimization problem that the ISAC RSU solves to find the optimal beam splitting ratio \(\boldsymbol{\beta}_t^*\) for fixed and state dependent settings accordingly is given by \eqref{eq:ISAC_AB_Objective}:
\begin{equation}
    \begin{aligned}\label{eq:ISAC_AB_Objective}
    \textbf{Fixed Utility: }&\boldsymbol{\beta}_t^* = \arg\max_{\boldsymbol{\alpha}_t^T\boldsymbol{\beta}_t\le 1}U(\boldsymbol{\beta}_t),\\
    \textbf{State Dependent Utility: }&\boldsymbol{\beta}_t^* = \arg\max_{\boldsymbol{\alpha}_t^T\boldsymbol{\beta}_t\le 1}U(\boldsymbol{\alpha}_t,\boldsymbol{\beta}_t).
\end{aligned}
\end{equation}
The constraints in \eqref{eq:ISAC_AB_Objective} ensure that when a vehicle \( m \) is accurately tracked by the RSU (i.e., perfect beam alignment can be assured), \( \beta_{t,m} \) is minimized to allocate a longer time duration to the high data rate narrow beam, thereby maximizing the total data throughput. Meanwhile, the objective functions in \eqref{eq:ISAC_AB_Objective} balance the trade-off between sensing accuracy and throughput by preventing \( \beta_{t,m} \) from becoming too small \cite{Liu2023ISAC_VN_Chapter}.
The utility function $U(\cdot)$ in \eqref{eq:ISAC_AB_Objective} characterizes the intent of the RSU, and is typically assumed to be \emph{concave} and (componentwise) \emph{monotonically increasing} in $\boldsymbol{\beta}$ so as to reflect saturation effects in performance gains (i.e., additional resource/waveform refinement yields progressively smaller improvements) while ensuring that better QoS performance is always preferred.
Accordingly, characterization of the RSU utility must satisfy both concavity and monotonicity with respect to the decision variable \(\boldsymbol{\beta}\).

The exact forms of the fixed and state dependent ISAC-AB utility functions are discussed in Sec.~\ref{subsec:PredSchedule}.
\section{The Proposed Inverse Learning Methodologies}\label{sec:ProposedMethod}
\begin{figure}[t]
    \centering
    \includegraphics[width=0.8\linewidth]{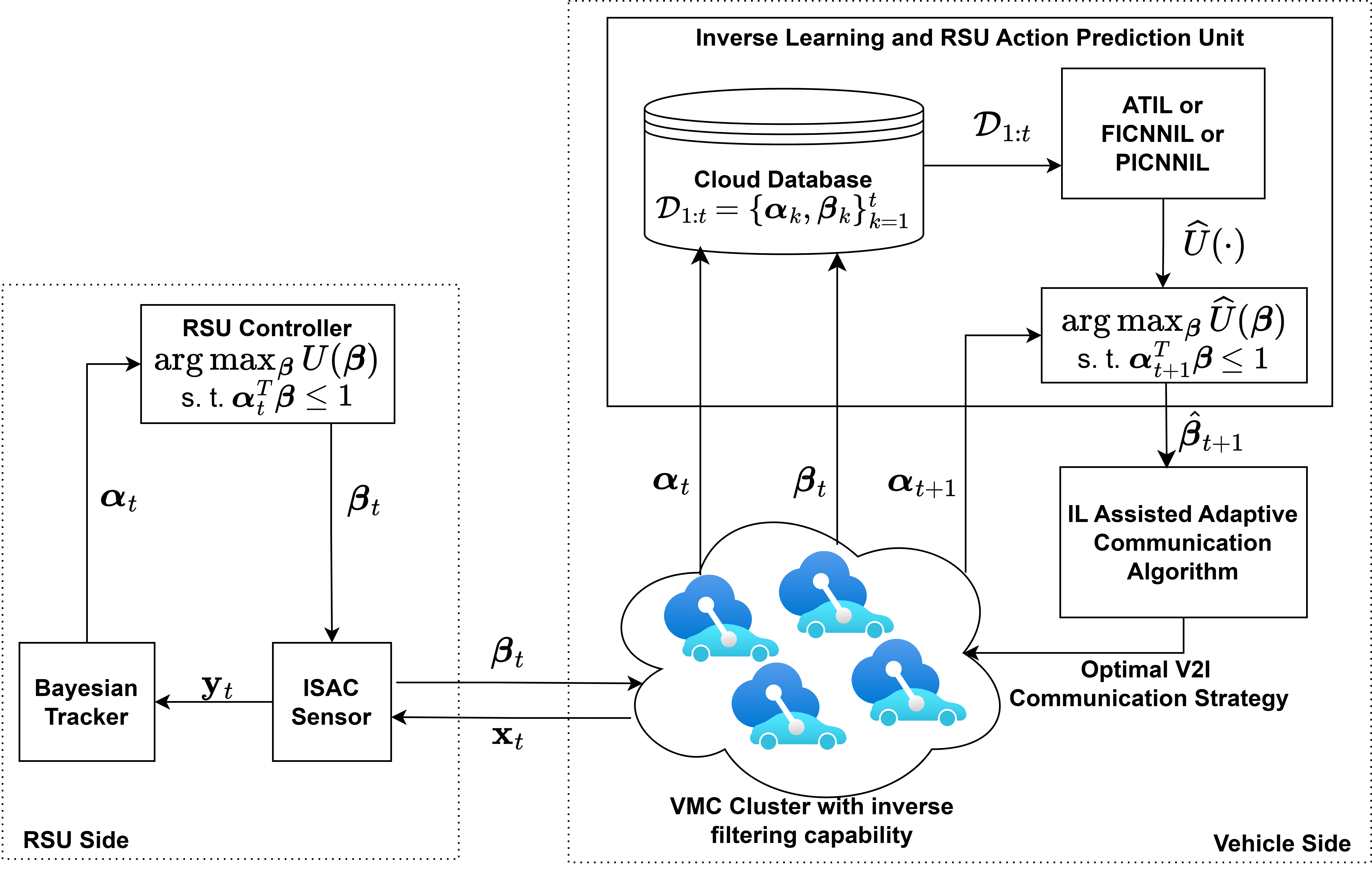}
    \caption{The proposed IL assisted adaptive communication strategies in VMC. ISAC RSU uses, \(\boldsymbol{\alpha}_t\), the tracking accuracy, of the Bayesian tracker to adaptively allocate ISAC-AB beam ratio (see Fig.~\ref{RSU_V2I}), which is unknown to VMC. The inverse learning procedure employed by the VMC aims to infer the unknown ISAC-AB beam switching strategy of the RSU.}
    \label{fig:BT_ATIL}
\end{figure}

\subsection{Learning Framework}\label{subsec:learningFW}
Fig.~\ref{fig:BT_ATIL} illustrates the proposed inverse learning framework for IL-assisted adaptive communication. 
At time \(t\), given the vehicle states \(\mathbf{x}_t=[\mathbf{X}_{t,1},\mathbf{X}_{t,2},\cdots,\mathbf{X}_{t,M}]^{T}\) and the corresponding RSU response \(\boldsymbol{\beta}_t\), the vehicles in the VMC employ a Bayesian inverse filtering procedure~\cite{krishnamurthy2019calibrate} to estimate the accuracy \(\boldsymbol{\alpha}_t\) of the RSU's Bayesian tracker. Over time, this yields a dataset \(\mathcal{D}_{1:t}\) consisting of the sequence of inferred sensing accuracies \(\boldsymbol{\alpha}_t\), together with the corresponding RSU beam allocation actions \(\boldsymbol{\beta}_t\). In the centralized learning paradigm, such data are assumed to be collected and stored in a cloud database for model training as shown in Fig.~\ref{fig:BT_ATIL}, whereas in the federated learning paradigm, the corresponding data remain distributed across different VMCs as shown in Fig.~\ref{fig:BT_ATIL_Fed}, and are used for local training without direct sharing of raw observations.
Based on \(\mathcal{D}_{1:t}\), the inverse learning algorithm aims to reconstruct the RSU strategy. Since the RSU intent, driven by the CSP, is modeled as a constrained utility maximization problem in~\eqref{eq:ISAC_AB_Objective}, inferring the underlying utility function \(U(\cdot)\) effectively characterizes the corresponding RSU decision-making strategy. The inferred utility can be subsequently exploited to design adaptive communication strategies for VMCs as presented in Sections~\ref{subsec:PredSchedule} and \ref{sec:ILDynamic}. 
Section~\ref{subsec:ATIL} describes ATIL, which employs revealed preference theory as a framework for non-parametric utility learning. 

\subsection{Revealed Preference and Afriat's Theorem based Inverse Learning (ATIL) - A Nonparametric Framework}\label{subsec:ATIL}
\begin{algorithm}[h]
{\fontsize{9pt}{9pt}\selectfont
\caption{ATIL Algorithm}\label{ATIL_Alg}
\begin{algorithmic}[1]
\REQUIRE $\mathcal{D}_{1:T}$
\STATE Generate Afriat's inequality as in \eqref{eq:AfrIneq}
\STATE Solve for $\mathcal{U}_{1:T} = \{U_t\}_{t=1}^T$ and ${\Lambda}_{1:T} = \{\lambda_t\}_{t=1}^T$
\STATE Use \(\mathcal{U}_{1:T}\) and \({\Lambda}_{1:T}\) in \eqref{eq:reconUtil} to reconstruct RSU utility $\hat{U}(\boldsymbol{\beta})$
\RETURN $\hat{U}(\boldsymbol{\beta})$
\end{algorithmic}}
\end{algorithm}
Revealed preference, a concept originally developed in microeconomics, is a framework for non-parametric learning of utility from observations of utility-maximizing behavior~\cite{varian1982nonparametric}. The celebrated Afriat's theorem~\cite{varian1982nonparametric} provides \begin{inparaenum}[(i)]
    \item necessary and sufficient conditions for detecting utility maximizing behavior,
    \item reconstructing utility function using piecewise linear functions.
\end{inparaenum}
\begin{theorem}[Afriat's Theorem\cite{AfriatsThm}]: 
Given a dataset 
\begin{equation}\label{eq:DataSet}
    \mathcal{D}_{1:T} = \{(\boldsymbol{\alpha}_t,\boldsymbol{\beta_t})\}_{t=1}^T,
\end{equation}
comprising observations of input-output behavior of an agent, the following statements are equivalent:
\begin{enumerate}
    \item The agent is a utility maximizer and there exists a monotonically increasing and concave utility function $U(\boldsymbol{\beta})$, which satisfies \eqref{eq:ISAC_AB_Objective}. 
    \item For $U_t$ and $\lambda_t>0$, the following set of inequalities has a feasible solution:
    \begin{equation}\label{eq:AfrIneq}
        U_s - U_t - \lambda_t\boldsymbol{\alpha}_t^T(\boldsymbol{\beta}_s - \boldsymbol{\beta}_t) \le 0 \forall t,s \in \{1,2,\cdots,T\}.
    \end{equation}
    \item A monotone and concave utility function that satisfies~\eqref{eq:ISAC_AB_Objective} is given by:
    \begin{equation}\label{eq:reconUtil}
       \hat{U}(\boldsymbol{\beta}) = \min_{t\in \{1,2,\cdots,T\}}\{U_t + \lambda_t\boldsymbol{\alpha}_t^T(\boldsymbol{\beta} - \boldsymbol{\beta}_t)\}.
    \end{equation}
\end{enumerate}\label{Th:Afriats}
\end{theorem}

Given the observations, \(\mathcal{D}_{1:T} = \{(\boldsymbol{\alpha}_t,\boldsymbol{\beta}_t)\}_{t=1}^T\), corresponding to a given RSU, our non-parametric approach for inverse learning the utility, ATIL, given in Algorithm~\ref{ATIL_Alg}, directly follows Theorem~\ref{Th:Afriats}. The fixed utility estimate, \(\hat{U}(\boldsymbol{\beta})\), learned using ATIL can be used to design predictive V2I communication protocols for vehicular networks. 

However, ATIL is a non-parametric learning approach, and the entire training dataset must be available to perform inference.
Therefore, Sec.~\ref{subsec:FICNNIL/PICNNIL_Parametric} presents parametric IL approaches, where, unlike ATIL, the ISAC RSU utility is modeled using neural networks and the parameters of the neural network are learned from the available observations \(\mathcal{D}_{1:t}\). However, the challenge lies in incorporating concavity and monotonicity in neural network and training the neural network given the observations $\mathcal{D}_{1:T}$.

\subsection{Parametric Inverse Networks - FICNN and PICNN}\label{subsec:FICNNIL/PICNNIL_Parametric}

The parametric representation used to characterize the RSU utility must satisfy both concavity and monotonicity with respect to the decision variable \(\boldsymbol{\beta}\).
Two neural network models available in the literature that enforce these structural properties are:
\begin{inparaenum}[(i)]
    \item Fully Input Concave Neural Network (FICNN)~\cite{amos2017input_ICNN_Seminal,grzeskiewicz2025uncovering_ICNN_Utility_Max}: a neural network whose input is the RSU decision variable \(\boldsymbol{\beta}_t\), and whose scalar output is concave with respect to the input. Hence, it is suitable for modeling fixed RSU utility, \(U(\boldsymbol{\beta}_t)\);
    \item Partially Input Concave Neural Network (PICNN)~\cite{amos2017input_ICNN_Seminal}: a neural network whose inputs are the state \(\boldsymbol{\alpha}_t\) and the RSU decision variable \(\boldsymbol{\beta}_t\), and whose scalar output is concave with respect to \(\boldsymbol{\beta}_t\) alone. Hence, it is suitable for modeling state-dependent RSU utility, \(U(\boldsymbol{\alpha}_t,\boldsymbol{\beta}_t)\).
\end{inparaenum} The methods adopted in the literature for training FICNN and PICNN~\cite{amos2017input_ICNN_Seminal} mainly focus on reliably reconstructing ordinal preferences alone. Hence, in this paper, we design a bilevel training approach that can also achieve low-MSE predictions of the RSU decision variable \(\boldsymbol{\beta}_t\) corresponding to the state \(\boldsymbol{\alpha}_t\).
\begin{figure}[t]
    \centering
    \includegraphics[width=0.8\linewidth]{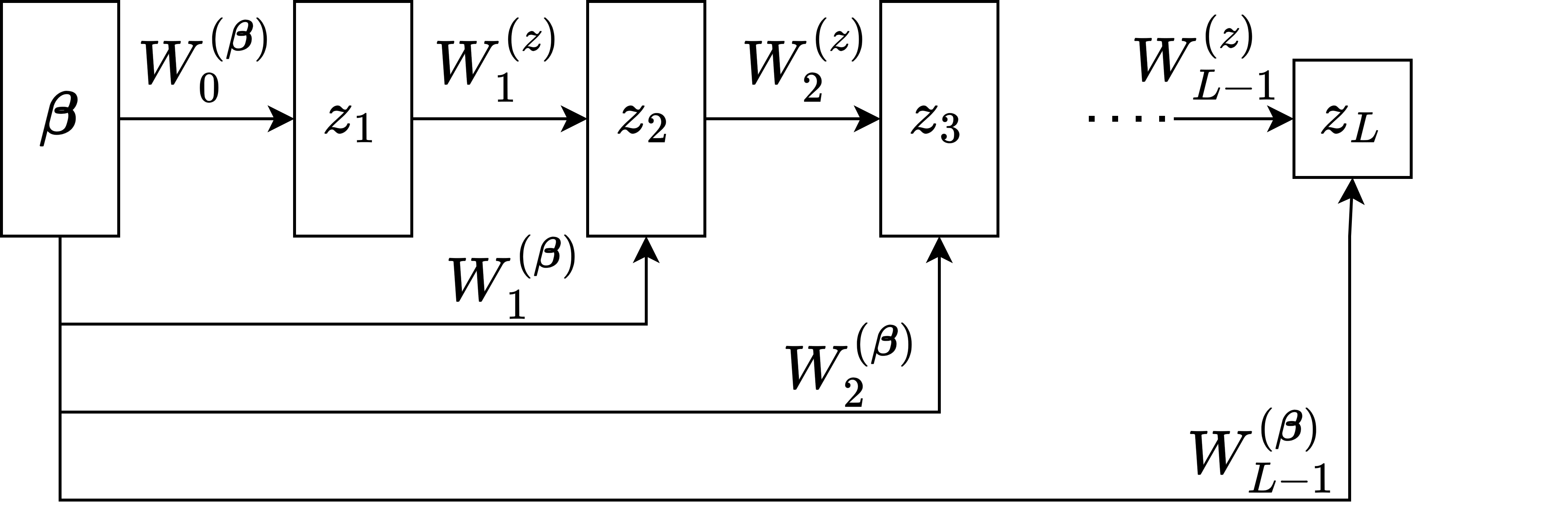}
\caption{ICNN architecture~\cite{amos2017input_ICNN_Seminal}, which takes the resource-allocation vector \(\boldsymbol{\beta}\) as input and outputs a scalar utility value \(U(\boldsymbol{\beta})\). The imposed constraints on the network weights, together with the choice of activation functions, ensure that \(U(\boldsymbol{\beta})\) is concave and monotonically increasing with respect to \(\boldsymbol{\beta}\).}
    \label{fig:ICNNArchi}
\end{figure}

In the remainder of this subsection, Sec.~\ref{subsec:ICNNArchitecture} presents the FICNN architecture together with the parameter constraints required to ensure concavity with respect to \(\boldsymbol{\beta}\). Sec.~\ref{subsec:PICNN_Archi} then presents the PICNN architecture and the associated parameter constraints required to guarantee concavity in \(\boldsymbol{\beta}\). Finally, Sec.~\ref{subsec:ICNNTraining} describes the bilevel optimization framework used to train both FICNN and PICNN based on revealed preference principles, and presents the corresponding parametric inverse learning algorithms, namely FICNNIL and PICNNIL.

\subsubsection{Fully Input Concave Neural Network (FICNN) Architecture}
\label{subsec:ICNNArchitecture}
We consider a fully connected, fully input-concave neural network with \(L\) layers, where each hidden layer has \(M\) neurons, defined over the input \(\boldsymbol{\beta} \in \mathbb{R}_+^M\), as illustrated in Fig.~\ref{fig:ICNNArchi}. 
For layers indexed by \(l = 0,1,\dots,L-1\), the network is defined recursively as \cite{amos2017input_ICNN_Seminal}:
\begin{equation}
\label{eq:ICNNArchi}
z_{l+1}
=
h_l\!\left(
W_l^{(z)} z_l
+
W_l^{(\boldsymbol{\beta})} \boldsymbol{\beta}
+
b_l
\right),
\;
\widehat{U}(\boldsymbol{\beta};\theta)
=
z_L,
\end{equation}
where
\begin{itemize}
    \item \(z_l\) denotes the activation vector at layer \(l\), with initialization \(z_0 \equiv 0\);
    \item \(W_l^{(z)} \in \mathbb{R}_+^{M\times M}\) and \(W_l^{(\boldsymbol{\beta})}\in \mathbb{R}^{M\times M}\) are weight matrices. For the last layer, \(W_{L-1}^{(z)} \in \mathbb{R}_+^{1\times M}\) and \(W_{L-1}^{(\boldsymbol{\beta})}\in \mathbb{R}^{1\times M}\);
    \item \(b_l \in \mathbb{R}^{M}\) are bias vectors. For the last layer, \(b_{L-1} \in \mathbb{R}\);
    \item \(h_l(\cdot)\) denotes an element-wise nonlinear activation function; and
    \item \(\theta = \{W_{0:L-1}^{(z)}, W_{0:L-1}^{(\boldsymbol{\beta})}, b_{0:L-1}\}\) denotes the collection of all trainable parameters.
\end{itemize}

\begin{proposition}\label{prop:Proposition_Concave}
The function \(\widehat{U}(\boldsymbol{\beta};\theta)\) in \eqref{eq:ICNNArchi} is concave in \(\boldsymbol{\beta}\) if the weight matrices \(W_l^{(z)}\) are element-wise non-negative for all \(l = 0,\dots,L-1\), and the activation functions \(h_l\) are concave and non-decreasing \cite{amos2017input_ICNN_Seminal,grzeskiewicz2025uncovering_ICNN_Utility_Max}.
\end{proposition}
The proof of Proposition~\ref{prop:Proposition_Concave} follows from the fact that non-negative sums of concave functions are also concave \cite{amos2017input_ICNN_Seminal}. 
The choice of non-linear concave activation functions \(h_l(\cdot)\) are:
\paragraph*{Concave-tanh}
Adapting the hyperbolic tangent, the `concave-tanh' activation function is given by\cite{grzeskiewicz2025uncovering_ICNN_Utility_Max}
\begin{equation}
h(x) =
\begin{cases}
\tanh(x), & x \ge 0, \\
x, & x < 0.
\end{cases}
\label{eq:concave_tanh}
\end{equation}

\paragraph*{Concave-sigmoid}
Adapting the sigmoid function, the `concave-sigmoid' activation function is given by\cite{grzeskiewicz2025uncovering_ICNN_Utility_Max}
\begin{equation}
h(x) =
\begin{cases}
\dfrac{1}{1+e^{-x}}, & x \ge 0, \\[6pt]
\dfrac{1}{4}x + \dfrac{1}{2}, & x < 0.
\end{cases}
\label{eq:concave_sigmoid}
\end{equation}

\paragraph*{Concave-log}
Adapting the natural logarithm, the `concave-log' activation function (with $\delta>0$ to avoid the asymptote at $x=0$) is given by\cite{grzeskiewicz2025uncovering_ICNN_Utility_Max}
\begin{equation}
h(x) =
\begin{cases}
\ln(x + \delta), & x > 0, \\[6pt]
\dfrac{1}{\delta}x + \ln \delta, & x \le 0.
\end{cases}
\label{eq:concave_log}
\end{equation}As in \cite{grzeskiewicz2025uncovering_ICNN_Utility_Max}, we use an FICNN with \(3\) hidden layers and \(M\) neurons per hidden layer, with \textit{concave-tanh} activation function to characterize the RSU utility. 

Although an FICNN can capture the concavity and monotonicity of the RSU utility with respect to \(\boldsymbol{\beta}_t\), it cannot explicitly capture the dependence of the utility on \(\boldsymbol{\alpha}_t\) in a state-dependent utility setting.

\subsubsection{Partially Input Concave Neural Network (PICNN) Architecture}\label{subsec:PICNN_Archi}
A PICNN~\cite{amos2017input_ICNN_Seminal} shown in Fig.~\ref{fig:PICNN_archi}, which combines a fully input-concave neural network that captures concavity and monotonicity with respect to \(\boldsymbol{\beta}\) with a conventional feedforward neural network that captures the dependence of the utility on \(\boldsymbol{\alpha}\), is better suited for characterizing RSU utility when the utility is state-dependent.
\begin{figure}[t]
    \centering
    \includegraphics[width=0.8\linewidth]{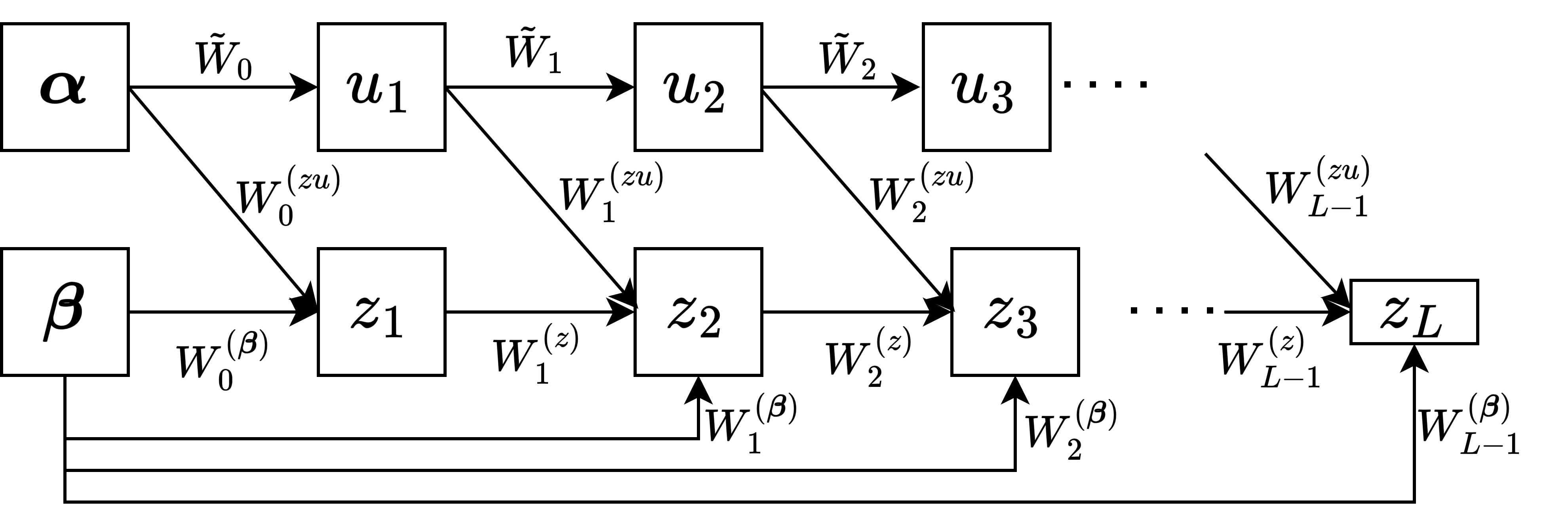}
\caption{PICNN architecture~\cite{amos2017input_ICNN_Seminal} for learning a state-dependent utility function \(U(\boldsymbol{\alpha},\boldsymbol{\beta})\). The network models dependence on the state vector \(\boldsymbol{\alpha}\), while maintaining concavity and monotonicity with respect to the resource-allocation vector \(\boldsymbol{\beta}\) through suitable constraints on the weights and activation functions.}

    \label{fig:PICNN_archi}
\end{figure}

Similar to FICNN in \eqref{eq:ICNNArchi}, PICNN is also a fully connected scalar valued neural network with \(L\) layers, defined over the inputs \(\boldsymbol{\alpha},\boldsymbol{\beta} \in \mathbb{R}_+^M\), as illustrated in Fig.~\ref{fig:PICNN_archi}. 
For layers indexed by \(l = 0,1,\dots,L-1\), the network is defined recursively as \cite{amos2017input_ICNN_Seminal}:
    \begin{align}
        u_{l+1} &= \tilde{h}_l(\tilde{W}_lu_l + \tilde{b}_l)\notag\\
        z_{l+1} &= h_l\left[ W_l^{(z)}\left(z_l \circ [W_l^{(zu)}u_l + b_l^{(z)}]\right)\right. +\\
        &\left. W_l^{(\boldsymbol{\beta})}\left(\boldsymbol{\beta}\circ(W_l^{(\boldsymbol{\beta}u)}u_l + b_l^{(\boldsymbol{\beta})})\right) + W_l^{(u)} + b_l \right],\notag\\
        \widehat{U}(\boldsymbol{\alpha},\boldsymbol{\beta};\theta) &= z_L,\quad u_0 = \boldsymbol{\alpha}.\notag\label{eq:PICNN_Archi}
    \end{align}
where
\begin{itemize}
    \item \(u_l\) and \(z_l\) denote the activation vectors of the \(u\)-path and \(z\)-path, respectively, at layer \(l\), with initialization \(u_0=\boldsymbol{\alpha}\) and \(z_0\equiv 0\). In \eqref{eq:PICNN_Archi}, \(\circ\) represents the Hadamard product, namely, elementwise multiplication between vectors of the same dimension.
    \item \(\tilde{W}_l \in \mathbb{R}^{M \times M}\) and \(\tilde{b}_l \in \mathbb{R}^{M}\) denote, respectively, the weight matrix and bias vector associated with the \(u\)-path,
    \item \(W_l^{(z)}\in \mathbb{R}_+^{M\times M}\), \(W_l^{(zu)} \in \mathbb{R}^{M\times M}\), \(W_l^{(\boldsymbol{\beta})}\in \mathbb{R}^{M\times M}\), \(W_l^{(\boldsymbol{\beta}u)}\in \mathbb{R}^{M\times M}\), and \(W_l^{(u)}\in \mathbb{R}^{M\times M}\) denote the trainable weight matrices (or vectors, as dictated by dimensional consistency) associated with the corresponding interactions in the PICNN architecture,
    \item \(b_l^{(z)}\in \mathbb{R}^{M}\), \(b_l^{(\boldsymbol{\beta})}\in \mathbb{R}^M\), and \(b_l\in \mathbb{R}^M\) denote the corresponding bias vectors,
    \item \(\tilde{h}_l(\cdot)\) and \(h_l(\cdot)\) denote element-wise nonlinear activation functions in the \(u\)-path and \(z\)-path, respectively, and
    \item \(\theta\) denotes the collection of all trainable parameters of the PICNN, i.e.,
    \begin{align}
\theta = \Big\{&
\tilde{W}_{0:L-1}, \tilde{b}_{0:L-1},
W_{0:L-1}^{(z)}, W_{0:L-1}^{(zu)},
W_{0:L-1}^{(\boldsymbol{\beta})},\notag\\
& W_{0:L-1}^{(\boldsymbol{\beta}u)}, 
W_{0:L-1}^{(u)},
b_{0:L-1}^{(z)}, b_{0:L-1}^{(\boldsymbol{\beta})}, b_{0:L-1}
\Big\}.
\end{align}
\end{itemize}

\begin{proposition}\label{propo:PICNN}
A PICNN with \(L\) layers can represent a FICNN (\eqref{eq:ICNNArchi}) with \(L\) layers, concave in \(\boldsymbol{\beta}\).
\end{proposition}
The proof of Proposition~\ref{propo:PICNN} follows directly from \cite{amos2017input_ICNN_Seminal}:
\begin{proof}
To recover an FICNN as a special case of the PICNN architecture, it is sufficient to eliminate the influence of the \(u\)-path by setting the corresponding weights to zero and choosing the associated bias terms so that the multiplicative gates reduce to constants. Under this choice, the PICNN reduces to an FICNN defined only in terms of \(\boldsymbol{\beta}\)\cite{amos2017input_ICNN_Seminal}.

\end{proof}
The nonlinear activation function \(h_l(\cdot)\) in the \(z\)-path of~\eqref{eq:PICNN_Archi}, analogous to the FICNN case, can be selected from the class of concave activation functions given in~\eqref{eq:concave_tanh}--\eqref{eq:concave_log}. In this work, we employ a PICNN with \(3\) hidden layers and \(M\) neurons in each hidden layer of both the \(u\)-path and the \(z\)-path. Further, the activation function \(h_l(\cdot)\) in~\eqref{eq:PICNN_Archi} is chosen as the \emph{concave-tanh} function.

\subsection{Parametric Inverse Learning FICNNIL and PICCNIL: Training Framework}\label{subsec:ICNNTraining}
Although the FICNN architecture in~\eqref{eq:ICNNArchi} and the PICNN architecture in~\eqref{eq:PICNN_Archi} are structurally different, their training procedures are conceptually similar. In this subsection, we describe a common training framework for both models, designed to enforce structural consistency with the constrained utility-maximizing behavior characterized by~\eqref{eq:ISAC_AB_Objective} using the revealed preference observations \(\mathcal{D}_{1:T}\) in~\eqref{eq:DataSet}.

Specifically, the training procedure enforces that the observed beam allocation decisions \(\boldsymbol{\beta}_t\), for \(t=1,\dots,T\), be consistent with the first-order optimality conditions associated with utility maximization over the corresponding feasible set defined by the resource constraint \(\boldsymbol{\alpha}_t^{\top}\boldsymbol{\beta}\le 1\) in~\eqref{eq:ISAC_AB_Objective}. For the FICNN model, the learned utility depends only on the action variable, i.e., \(U(\boldsymbol{\beta})\), whereas for the PICNN model, the learned utility is state-dependent and is represented as \(U(\boldsymbol{\alpha},\boldsymbol{\beta})\). Nevertheless, in both cases, the training objective is to learn model parameters such that the predicted optimizer induced by the learned utility matches the observed behavior.

{\sl \textbf{The proposed bilevel training framework:} }
To estimate the FICNN/PICNN parameters \(\theta\), we adopt a bilevel optimization
framework inspired by~\cite{grzeskiewicz2025uncovering_ICNN_Utility_Max,amos2017input_ICNN_Seminal}.
PEARL~\cite{grzeskiewicz2025uncovering_ICNN_Utility_Max} is well suited to
econometric settings in which expenditures vary across observations and the
money-metric expenditure carries informative revealed-preference signal.
In our normalized ISAC RSU allocation problem, however, the budget constraint is
active and the expenditure is fixed, i.e.,
\(\boldsymbol{\alpha}^{T}\boldsymbol{\beta}^{\star}=1\), so matching expenditure
provides little information about the component-wise allocation vector.
Since multiple allocations can satisfy the same unit expenditure while inducing
different beam splitting outcomes, the proposed bilevel ICNN
training directly minimizes the predicted--observed allocation error and is
therefore better aligned with component level reconstruction of \(\boldsymbol{\beta}\). See Table~\ref{tab:il_comparison_t20_t200} in Sec.~\ref{sec:NumRes} for a comparison between PEARL and the proposed bilevel training framework in ISAC-RSU context.

For each revealed preference observation \((\boldsymbol{\alpha}_t,\boldsymbol{\beta}_t)\in\mathcal{D}_{1:T}\), the learning problem is formulated as follows.
For the FICNN model, the bilevel problem is given by
\begin{align}
    \textbf{Inner problem:}\;&
    \boldsymbol{\beta}_{\theta}(\boldsymbol{\alpha}_t)
    =
    \arg\max_{\boldsymbol{\alpha}_t^{\top}\boldsymbol{\beta}\le 1}
    \widehat{U}(\boldsymbol{\beta};\theta),
    \label{eq:BilevelProblemFICNN}\\
    \textbf{Outer problem:}\;&
    \theta
    =
    \arg\min_{\theta}\mathcal{L}_{\theta}
    =
    \sum_{\forall t}
    \left\|
    \boldsymbol{\beta}_{\theta}(\boldsymbol{\alpha}_t)-\boldsymbol{\beta}_t
    \right\|^2.
    \notag
\end{align}

For the PICNN model, the corresponding bilevel problem is
\begin{align}
    \textbf{Inner problem:}\;&
    \boldsymbol{\beta}_{\theta}(\boldsymbol{\alpha}_t)
    =
    \arg\max_{\boldsymbol{\alpha}_t^{\top}\boldsymbol{\beta}\le 1}
    \widehat{U}(\boldsymbol{\alpha}_t,\boldsymbol{\beta};\theta),
    \label{eq:BilevelProblemPICNN_Common}\\
    \textbf{Outer problem:}\;&
    \theta
    =
    \arg\min_{\theta}\mathcal{L}_{\theta}
    =
    \sum_{\forall t}
    \left\|
    \boldsymbol{\beta}_{\theta}(\boldsymbol{\alpha}_t)-\boldsymbol{\beta}_t
    \right\|^2.
    \notag
\end{align}
The inner problem computes the optimal beam splitting vector predicted by the FICNN/PICNN utility model under the resource constraint \(\boldsymbol{\alpha}_t^\top\boldsymbol{\beta}\le 1\), while the outer problem updates the corresponding model parameters so that the predicted RSU actions match the observed revealed preference data.
\begin{algorithm}[t]
{\fontsize{9pt}{9pt}\selectfont
\caption{Differentiable projected gradient ascent algorithm for solving inner optimization problem}
\label{alg:InnerProblem_Common}
\begin{algorithmic}[1]

\REQUIRE $\mathcal{D}_{1:T} = \{(\boldsymbol{\alpha}_t, \boldsymbol{\beta}_t)\}_{t=1}^{T}$, 
FICNN/PICNN parameters $\theta$, inner steps $J$, step size $\eta_{in}$

\STATE Initialize $\{\boldsymbol{\beta}^{(0)}_t\}_{t=1}^{T}$

\FOR{$j=0$ to $J-1$}
    \FOR{$t=1$ to $T$}
        \STATE Compute utility gradient:
        \(
        \mathbf{g}^{(j)}_t 
        =
        \nabla_{\boldsymbol{\beta}} 
        \widehat{U}_t\!\left(\boldsymbol{\beta}^{(j)}_t;\theta\right)
        \)
        \STATE Gradient ascent:
        \(
        \tilde{\boldsymbol{\beta}}^{(j+1)}_t
        =
        \boldsymbol{\beta}^{(j)}_t 
        + 
        \eta_{in} \mathbf{g}^{(j)}_t
        \)
        \STATE Projection into feasible region:
        \[
        \boldsymbol{\beta}^{(j+1)}_t
        =
        \Pi_{\{\boldsymbol{\alpha}_t^\top \boldsymbol{\beta} \le 1,\ \boldsymbol{\beta}\ge0\}}
        \!\left(\tilde{\boldsymbol{\beta}}^{(j+1)}_t\right)
        \]
    \ENDFOR
\ENDFOR

\RETURN Approximate solutions of inner optimization problem:
\[
\boldsymbol{\beta}_{\theta}(\boldsymbol{\alpha}_t)
=
\boldsymbol{\beta}^{(J)}_t,
\quad t=1,\dots,T
\]
\end{algorithmic}}
\end{algorithm}

{\sl \textbf{Solving the outer problem:} }
The outer optimization problem can be solved via classical regression, where the FICNN/PICNN parameters are learned through backpropagation which iteratively update \(\theta\) so as to minimize the outer-level loss function \(\mathcal{L}_{\theta}\) defined in~\eqref{eq:BilevelProblemFICNN} or~\eqref{eq:BilevelProblemPICNN_Common}. The total gradient of the loss with respect to \(\theta\) is given by
\begin{equation}
\label{eq:LossGradientTotal_Common}
\nabla_{\theta}\mathcal{L}_{\theta}
=
\frac{1}{T}\sum_{t=1}^{T}
\frac{\partial \mathcal{L}_{\theta}}{\partial \boldsymbol{\beta}_{\theta}(\boldsymbol{\alpha}_t)}
\;
\frac{\partial \boldsymbol{\beta}_{\theta}(\boldsymbol{\alpha}_t)}{\partial \theta},
\end{equation}
where \(\boldsymbol{\beta}_{\theta}(\boldsymbol{\alpha}_t)\) denotes the solution of the corresponding inner optimization problem. For the PICNN case, although the gradient is taken with respect to \(\theta\), it depends implicitly on \(\boldsymbol{\alpha}_t\) through the inner-level solution, since both the utility \(\widehat{U}(\boldsymbol{\alpha}_t,\boldsymbol{\beta};\theta)\) and the feasible set are parameterized by \(\boldsymbol{\alpha}_t\). However, no derivative with respect to \(\boldsymbol{\alpha}_t\) appears explicitly, since \(\boldsymbol{\alpha}_t\) is treated as observed data during training.

{\sl \textbf{Solving inner problem:} }Since the mapping \(\theta \mapsto \boldsymbol{\beta}_{\theta}(\boldsymbol{\alpha}_t)\) is implicitly defined through the inner maximization in~\eqref{eq:BilevelProblemFICNN} and~\eqref{eq:BilevelProblemPICNN_Common}, the gradient computation must differentiate through the inner optimization pipeline. To enable end-to-end differentiability, and following~\cite{grzeskiewicz2025uncovering_ICNN_Utility_Max}, we employ a differentiable projected gradient ascent scheme to solve the inner problem, as outlined in Algorithm~\ref{alg:InnerProblem_Common}. 
This allows the outer-level gradient to propagate through the iterative inner updates during training. 
Given an FICNN or PICNN representation of the utility, Algorithm~\ref{alg:InnerProblem_Common} iteratively computes the corresponding optimal \(\boldsymbol{\beta}\) by
\begin{inparaenum}[(i)]
    \item randomly initializing \(\boldsymbol{\beta}\);
    \item iteratively computing the gradient of the utility with respect to \(\boldsymbol{\beta}\) and updating \(\boldsymbol{\beta}\) using gradient ascent; and
    \item projecting \(\boldsymbol{\beta}\) back onto the feasible region prescribed by the constraints,
\end{inparaenum}
for a predefined number of iterations.

\begin{algorithm}[t]
{\fontsize{9pt}{9pt}\selectfont
\caption{FICNNIL/PICNNIL algorithm for inverse learning of RSU utility}
\label{alg:ICNNIL_Common}
\begin{algorithmic}[1]
\REQUIRE $\mathcal{D}_{1:T} = \{(\boldsymbol{\alpha}_t, \boldsymbol{\beta}_t)\}_{t=1}^{T}$, 
FICNN/PICNN network, number of iterations $I$, step size $\eta_{out}$

\STATE Initialize model parameters \(\theta\)

\FOR{iterations \(i = 0\) to \(I-1\)}
\STATE \(\forall t\), use Algorithm~\ref{alg:InnerProblem_Common} to solve
\[
\boldsymbol{\beta}_{\theta}(\boldsymbol{\alpha}_t)
=
\arg\max_{\boldsymbol{\alpha}_t^{\top}\boldsymbol{\beta}\le 1}
\widehat{U}_t(\boldsymbol{\beta};\theta)
\]

\STATE Compute batch outer loss:
\(
\mathcal{L}_{\theta}
=
\frac{1}{T}
\sum_{t=1}^{T}
\left\|
\boldsymbol{\beta}_{\theta}(\boldsymbol{\alpha}_t)
-
\boldsymbol{\beta}_t
\right\|^2
\)

\STATE Compute total derivative using~\eqref{eq:LossGradientTotal_Common}:
\[
\nabla_\theta \mathcal{L}_{\theta}
=
\frac{1}{T}
\sum_{t=1}^{T}
\frac{\partial \mathcal{L}_{\theta}}{\partial \boldsymbol{\beta}_{\theta}(\boldsymbol{\alpha}_t)}
\cdot
\frac{\partial \boldsymbol{\beta}_{\theta}(\boldsymbol{\alpha}_t)}{\partial \theta}
\]

\STATE Update parameters:
\(
\theta \leftarrow \theta - \eta_{out} \nabla_\theta \mathcal{L}_{\theta}
\)
\ENDFOR

\RETURN Updated \(\theta\), trained FICNN/PICNN
\[
\widehat{U}(\cdot;\theta):\boldsymbol{\beta}\rightarrow\mathbb{R}
\quad \text{or} \quad
\widehat{U}(\cdot,\cdot;\theta):(\boldsymbol{\alpha},\boldsymbol{\beta})\rightarrow\mathbb{R}
\]

\end{algorithmic}}
\end{algorithm}

Step~6 of Algorithm~\ref{alg:InnerProblem_Common} corresponds to the Euclidean projection of the intermediate iterate onto the convex feasible set
\[
\mathcal{C}_t
=
\left\{
\boldsymbol{\beta}\in\mathbb{R}_+^M
\;\middle|\;
\boldsymbol{\alpha}_t^\top\boldsymbol{\beta}\le 1
\right\},
\]
which can, in principle, be computed exactly by solving a convex projection problem. However, in our implementation, we use an approximate numerical procedure \cite{amos2017input_ICNN_Seminal,grzeskiewicz2025uncovering_ICNN_Utility_Max,parikh2014proximal} for computational simplicity. Specifically, the intermediate iterate, \(\tilde{\boldsymbol{\beta}}\) in step~5 is first clipped elementwise to enforce non-negativity,
\(
\bar{\boldsymbol{\beta}}
=
\left[
\tilde{\boldsymbol{\beta}}
\right]_+ .
\) Here \(\left[\cdot\right]_+\) represents elementwise non-negative clipping operation.
If the resulting vector violates the budget constraint, i.e.,
\(
\boldsymbol{\alpha}_t^\top\bar{\boldsymbol{\beta}}>1,
\)
it is corrected along the normal direction of the budget hyperplane
\(\boldsymbol{\alpha}_t^\top\boldsymbol{\beta}=1\) as
\[
\boldsymbol{\beta}'
=
\bar{\boldsymbol{\beta}}
-
\frac{
\boldsymbol{\alpha}_t^\top\bar{\boldsymbol{\beta}}-1
}{
\|\boldsymbol{\alpha}_t\|_2^2
}
\boldsymbol{\alpha}_t .
\]
Elementwise non-negativity is then enforced again if needed by setting
\(
\boldsymbol{\beta}^{+}
=
\left[
\boldsymbol{\beta}'
\right]_+ .
\)

Algorithm~\ref{alg:ICNNIL_Common} (FICNNIL/PICCNNIL) alternately solves the inner and outer optimization problems in \eqref{eq:BilevelProblemFICNN} or \eqref{eq:BilevelProblemPICNN_Common} in an iterative manner to learn the optimal FICNN/PICNN parameters that characterize the ISAC RSU behaviour consistent with \eqref{eq:ISAC_AB_Objective}.

\begin{figure}[t]
    \centering
    \includegraphics[width=\linewidth]{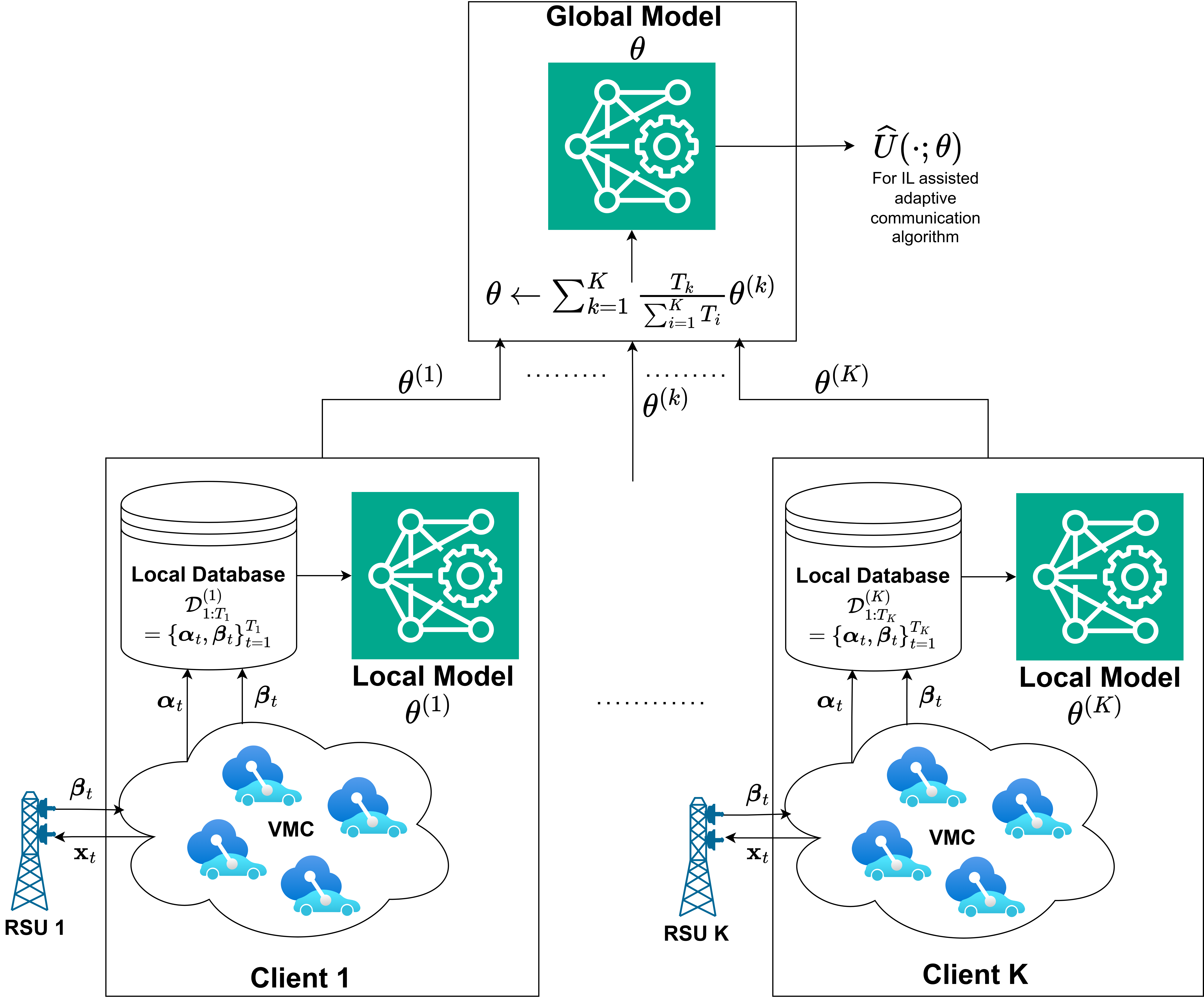}
    \caption{Illustration of federated averaging based inverse learning of the ISAC RSU utility. In the proposed framework, geographically separated VMCs (clients) observe different subsets of the revealed-preference dataset, represented as local observations \(\mathcal{D}^{(k)}_{1:T_k}=\{(\alpha_t,\beta_t)\}_{t=1}^{T_k}\), corresponding to the behaviour of respective serving RSUs.
    Individual VMC clients use the available local data to train client-specific local models \(\theta^{(k)}\) using either FICNNIL or PICNNIL. The resulting local parameter updates are communicated to a central server, which aggregates them through the federated averaging rule to obtain an updated global model \(\theta\). Repeating this procedure over multiple communication rounds yields the final global inverse learning model.}
    \label{fig:BT_ATIL_Fed}
\end{figure}
\subsection{Federated Inverse Learning: FedFICNNIL/FedPICNNIL}\label{subsec:FedInvLearning} 
When the utility is state-dependent, different RSUs deployed at geographically different locations, such as straight highways, junctions, and free roads, may perceive different levels of tracking uncertainty. Hence, the corresponding VMCs served by those RSUs may observe state-dependent behaviours that differ from one another. In such situations, efficient inverse learning requires data corresponding to different states of the world, i.e., the sensing accuracies along with the actions of all the RSUs. However, as shown in Fig.~\ref{fig:BT_ATIL_Fed}, the entire training data may not be available at a centralized cloud server, either due to privacy concerns among different VMCs or to avoid communication overhead when communication resources are scarce. For such decentralized settings, we propose federated inverse learning, where individual VMCs train parametric IL models locally using the available observations from their respective serving ISAC RSUs, and the corresponding parameter updates from all VMCs, referred to as clients in federated learning terminology, are aggregated to obtain the final model.
Since ATIL is non-parametric, it is not readily suitable for federated learning settings, but FICNNIL and PICNNIL are.

In this section, we extend the parametric IL algorithms, FICNNIL/PICNNIL in Algorithm~\ref{alg:ICNNIL_Common}, to a federated setting, as shown in Fig.~\ref{fig:BT_ATIL_Fed}. 
In contrast to the centralized learning paradigm discussed in Sec.~\ref{subsec:learningFW}, in the federated inverse learning framework illustrated in Fig.~\ref{fig:BT_ATIL_Fed}, geographically separated VMCs train local models, either FICNN or PICNN, using datasets of the form \eqref{eq:DataSet} that are locally available to them, i.e., the observations obtained from their respective serving RSUs. Subsequently, the VMCs communicate only model updates to a coordinating server, which aggregates them to construct a global model. 
In FedFICNNIL/FedPICNNIL, given in Algorithm~\ref{alg:Federated_ICNNIL}, each VMC acts as a client possessing a local revealed-preference dataset \(\mathcal{D}_{1:T_k}^{(k)}\). Each client uses Algorithm~\ref{alg:ICNNIL_Common} to update its local parameters over \(R\) communication rounds. In each communication round, the local model updates are shared with a coordinating server, and the global FICNN/PICNN model is collaboratively trained using the FedAvg algorithm~\cite{mcmahan2017FedAvg}.

\begin{algorithm}[t]
{\fontsize{9pt}{9pt}\selectfont
\caption{FedFICNNIL/FedPICNNIL algorithm for federated inverse learning of RSU utility using FedAvg}
\label{alg:Federated_ICNNIL}
\begin{algorithmic}[1]
\REQUIRE \(K\) VMC clients, local datasets \(\{\mathcal{D}_{1:T_k}^{(k)}\}_{k=1}^{K}\), FICNN/PICNN network, number of communication rounds \(R\), number of local iterations \(I_{\mathrm{loc}}\), outer step size \(\eta_{out}\)

\STATE Initialize global model parameters \(\theta^{(0)}\) at the server

\FOR{communication rounds \(r=0\) to \(R-1\)}
    \STATE Server broadcasts current global model \(\theta^{(r)}\) to all participating clients

    \FORALL{clients \(k=1,\dots,K\) \textbf{in parallel}}
        \STATE Initialize local parameters:
        \(
        \theta_k^{(r,0)} \gets \theta^{(r)}
        \)

        \STATE Use Algorithm~\ref{alg:ICNNIL_Common} over \(I_{\mathrm{loc}}\) iterations to update the local parameter of client \(k\) to get \(\theta_k^{(r,I_{\mathrm{loc}})}\).

        \STATE Client \(k\) sends updated local parameters \(\theta_k^{(r,I_{\mathrm{loc}})}\) to the server
    \ENDFOR

    \STATE Server aggregates local models using FedAvg\cite{mcmahan2017FedAvg}:
    \[
    \theta^{(r+1)}
    \gets
    \sum_{k=1}^{K}\frac{T_k}{\sum_{j=1}^{K}T_j}\,\theta_k^{(r,I_{\mathrm{loc}})}
    \]
\ENDFOR

\RETURN Global trained FICNN/PICNN model:
\[
\widehat{U}(\cdot;\theta^{(R)}):\boldsymbol{\beta}\rightarrow\mathbb{R}
\quad \text{or} \quad
\widehat{U}(\cdot,\cdot;\theta^{(R)}):(\boldsymbol{\alpha},\boldsymbol{\beta})\rightarrow\mathbb{R}
\]

\end{algorithmic}}
\end{algorithm}

\subsection{IL-assisted predictive scheduling for cooperative data downloading in VMC}\label{subsec:PredSchedule}
As an example use case, we first consider a cooperative data downloading framework in which mutually exclusive chunks of a common file, e.g., an HD map or a media file, are downloaded by different members of the VMC and subsequently reconstructed through V2V communication~\cite{survey2019coopDownlAdvant}.

To demonstrate inverse learning of RSU strategy, consider the ISAC RSU with adaptive beam forming capability\cite{Liu2023ISAC_VN_Chapter} as illustrated in Fig.~\ref{RSU_V2I}, working under a communication-centric CSP. The ISAC RSU tries to maximize the total V2I data rate, while ensuring reliable connectivity to all the vehicles in the VMC which it serves. 
Let \(R(W)\) and \(R(N)\) denote the maximum achievable data rates, i.e., the Shannon capacities, corresponding to widebeam and narrowbeam transmission, respectively. Further, for a beam splitting ratio of \(\beta_{t,m}\in [0,1]\) the effective narrowbeam data rate for vehicle \(m\) at time \(t\) is given by\cite{Liu2023ISAC_VN_Chapter}
\begin{equation}\label{eq:RNbeta}
    R(N,\beta_{t,m}) = P_A(\beta_{t,m})\,R(N),
\end{equation}
where
\begin{equation}\label{eq:PA}
    P_A(\beta_{t,m})
    =
    \operatorname{erf}\!\left(
    \sqrt{\beta_{t,m}/2}\;\delta_{t,m}\sigma_{\phi}
    \right)
\end{equation}denotes the beam alignment probability of the narrow beam, and \(\operatorname{erf}(\cdot)\) is the Gauss error function~\cite{Liu2023ISAC_VN_Chapter}. In~\eqref{eq:PA}, \(\delta_{t,m}\) and \(\sigma_\phi^2\) denote, respectively, the half-power beamwidth of the narrow beam and the noise variance in observing \(\phi\)~\cite{Liu2023ISAC_VN_Chapter}.

Given a beam-splitting ratio \(\beta_{t,m}\), the ISAC AB downlink data rate available for vehicle \(m\) in Fig.~\ref{RSU_V2I} is a weighted sum of two components, namely:
\begin{inparaenum}[(i)]
    \item the wide-beam data rate, \(\beta_{t,m}R(W)\); and
    \item the narrow-beam data rate, \((1-\beta_{t,m})R(N,\beta_{t,m})\), where \(R(N,\beta_{t,m})\) is given in \eqref{eq:RNbeta}.
\end{inparaenum}
By allocating \(\beta_{t,m}\) close to \(1\), the ISAC RSU can ensure highly reliable V2I data transfer at a low data rate, whereas choosing \(\beta_{t,m}\) close to \(0\) provides the highest possible V2I data throughput when the beam is perfectly aligned, but may result in packet loss when there is beam misalignment. However, the adaptive beam allocation by the ISAC RSU in \eqref{eq:ISAC_AB_Objective} depends on the RSU utility. A fixed utility may assign fixed priorities to reliability and maximum achievable throughput irrespective of the state \(\boldsymbol{\alpha}_t\), whereas, under a state-dependent utility, the priority between reliability and maximum achievable throughput is a function of \(\boldsymbol{\alpha}_t\), denoted by \(p(\boldsymbol{\alpha}_t)\in[0,1]\)~\cite{kuo2007utility_P_alpha,du2023v2i_isac_P_alpha}.

Specifically, the ISAC RSU is assumed to maximize a weighted sum-rate utility function~\cite{Liu2023ISAC_VN_Chapter}, given by
\begin{align}
    U(\boldsymbol{\alpha}_t,\boldsymbol{\beta}_t)
    =&
    \sum_{m=1}^{M}
    w_m
    \Big[
    \big(1-p(\boldsymbol{\alpha}_t)\big)\,\beta_{t,m}R(W)
    \notag\\
    &
    \;\;+
    p(\boldsymbol{\alpha}_t)
    (1-\beta_{t,m})R(N,\beta_{t,m})
    \Big].
    \label{eq:ISACABUtil}
\end{align}
The formulation in \eqref{eq:ISACABUtil}, accommodates both fixed and state dependent RSU utilities through the dependence of \(p(\boldsymbol{\alpha}_t)\) on the state \(\boldsymbol{\alpha}_t\). In particular,
\begin{equation}\label{eq:p_alpha}
    p(\boldsymbol{\alpha}_t)\equiv
    \begin{cases}
        p_0,\quad p_0\in[0,1], & \text{; fixed utility,}\\[0.5ex]
        f(\boldsymbol{\alpha}_t)\in[0,1], & \text{; state dependent utility,}
    \end{cases}
\end{equation}
where \(f(\cdot)\) is an arbitrary state-dependent mapping as used in utility based resource allocation\cite{kuo2007utility_P_alpha}. 
In general, the parameters \(w_m\), \(p(\boldsymbol{\alpha}_t)\), and \(\delta_{t,m}\), as well as the RSU utility, are unknown to the vehicles.

In the cooperative data downloading setting, an efficient V2I protocol should schedule the file chunks to be downloaded by different vehicles subject to the available data rate \(R_{t',m}\) offered by the serving RSU~\cite{electronics11223663PreCache}. However, such a strategy must be executed in a predictive manner, since download scheduling decisions need to be made in advance based on anticipated communication conditions and RSU behavior~\cite{electronics11223663PreCache,sohail2023routingProtoVanet}. Consequently, learning the RSU strategy is essential for optimal predictive scheduling. The inverse learning-assisted predictive scheduling framework for cooperative data downloading in VMC is presented in Algorithm~\ref{ATIL_Based_Pred_Scheduling}. Given the anticipated state \(\boldsymbol{\alpha}_{t'}\), the IL-assisted predictive scheduling algorithm (Algorithm~\ref{ATIL_Based_Pred_Scheduling}) uses the utility model already inverse-learned using ATIL, FICNNIL, PICNNIL, FedFICNNIL, or FedPICNNIL in \eqref{eq:ISAC_AB_Objective} to predict the corresponding RSU beam allocation \(\boldsymbol{\beta}_{t'}\). Based on the predicted beam-splitting ratio, the algorithm predicts the RSU-to-vehicle data rate available to individual vehicles in the VMC and schedules the data download accordingly.
\begin{algorithm}[t]
{\fontsize{9pt}{9pt}\selectfont
\caption{IL-assisted predictive scheduling for cooperative data downloading in VMC}\label{ATIL_Based_Pred_Scheduling}
\begin{algorithmic}[1]
\REQUIRE $\mathcal{D}_{1:T}$, $\boldsymbol{\alpha}_{t'}$, $R(N)$, $R(W)$, $\delta_{t',m}$, and $\sigma_{\phi}$

\STATE Learn the RSU utility model \(\widehat{U}\) using an IL algorithm, e.g., ATIL, FICNNIL, or PICNNIL
\STATE Predict the RSU beam allocation \(\boldsymbol{\beta}_{t'}\) corresponding to \(\boldsymbol{\alpha}_{t'}\):
\[
\boldsymbol{\beta}_{t'}
\gets
\arg\max_{(\boldsymbol{\alpha}_{t'})^T\boldsymbol{\beta}\le 1}
\widehat{U}_{t'}(\cdot)
\]
\FORALL{$m$}
\STATE Compute the predicted available data rate:
\[
R_{t',m}
\gets
\beta_{t',m}R(W)
+
(1-\beta_{t',m})R(N,\beta_{t',m})
\]
\STATE Schedule download of \(R_{t',m}\) bits per unit contact time
\ENDFOR
\end{algorithmic}}
\end{algorithm}

\subsection{IL-Assisted Dynamic Cluster Head Selection in VMC (ILDynamic)}\label{sec:ILDynamic}
The procedure outlined in Algorithm~\ref{ATIL_Based_Pred_Scheduling} can be readily extended to enable dynamic cluster head selection in the VMC framework, as presented in Algorithm~\ref{alg:ILDynamic}. Specifically, Step~5 of the algorithm is modified to identify the vehicle \(m\) that achieves the highest predicted communication data rate with the RSU, and to designate that vehicle as the cluster head. 
\begin{algorithm}[t]
{\fontsize{9pt}{9pt}\selectfont
\caption{IL-assisted dynamic cluster head selection in VMC (ILDynamic)}\label{alg:ILDynamic}
\begin{algorithmic}[1]
\REQUIRE $\mathcal{D}_{1:T}$, $\boldsymbol{\alpha}_{t'}$, $R(N)$, $R(W)$, $\delta_{t',m}$, and $\sigma_{\phi}$

\STATE Learn the RSU utility model \(\widehat{U}\) using an IL algorithm, e.g., ATIL, FICNNIL, or PICNNIL
\STATE Predict the RSU beam allocation \(\boldsymbol{\beta}_{t'}\) corresponding to \(\boldsymbol{\alpha}_{t'}\):
\[
\boldsymbol{\beta}_{t'}
\gets
\arg\max_{(\boldsymbol{\alpha}_{t'})^T\boldsymbol{\beta}\le 1}
\widehat{U}_{t'}(\cdot)
\]

\FORALL{$m$}
\STATE Compute the predicted available data rate:
\[
R_{t',m}
\gets
\beta_{t',m}R(W)
+
(1-\beta_{t',m})R(N,\beta_{t',m})
\]
\ENDFOR

\STATE Select the cluster head:
\(
m^\star \gets \arg\max_{m\in\{1,\dots,M\}} R_{t',m}
\)
\STATE Designate vehicle \(m^\star\) as the cluster head

\end{algorithmic}}
\end{algorithm}

The inverse learning of the ISAC-RSU strategy, together with the corresponding predictive scheduling and dynamic cluster head selection protocols under the proposed IL framework, is illustrated through detailed numerical examples in Sec.~\ref{sec:NumRes}.

\section{Numerical Results}\label{sec:NumRes}
To perform the numerical experiments, we consider the system model in Fig.~\ref{RSU_V2I}, consisting of a VMC with \(M=3\) vehicles within the footprint of an RSU. The numerical study focuses on the following two use cases: \begin{inparaenum}[(i)]
    \item adaptive cluster head selection in the VMC, and
    \item predictive scheduling for cooperative data downloading in the VMC.
\end{inparaenum}
Further, the numerical examples are used to evaluate the proposed inverse learning approaches, namely ATIL, FICNNIL, and PICNNIL.

\begin{table*}[t]
\centering
\caption{Comparison of inverse learning models for fixed and state dependent utility settings under different numbers of observations}
\label{tab:il_comparison_t20_t200}
\resizebox{\textwidth}{!}{%
\begin{tabular}{|l|cccc|cccc|cccc|cccc|}
\hline
\multirow{3}{*}{\textbf{Model}} 
& \multicolumn{8}{c|}{\textbf{$T = 20$}} 
& \multicolumn{8}{c|}{\textbf{$T = 200$}} \\ \cline{2-17}

& \multicolumn{4}{c|}{\textbf{Fixed Utility}} 
& \multicolumn{4}{c|}{\textbf{State Dependent Utility}} 
& \multicolumn{4}{c|}{\textbf{Fixed Utility}} 
& \multicolumn{4}{c|}{\textbf{State Dependent Utility}} \\ \cline{2-17}
& \textbf{MSE}\((\downarrow)\) 
& \textbf{\(\rho\)}\((\uparrow)\) 
& \textbf{\(\bar{R}\)}\((\uparrow)\) 
& \textbf{\(\bar{L}\)}\((\downarrow)\) 
& \textbf{MSE}\((\downarrow)\) 
& \textbf{\(\rho\)}\((\uparrow)\) 
& \textbf{\(\bar{R}\)}\((\uparrow)\) 
& \textbf{\(\bar{L}\)}\((\downarrow)\) 
& \textbf{MSE}\((\downarrow)\) 
& \textbf{\(\rho\)}\((\uparrow)\) 
& \textbf{\(\bar{R}\)}\((\uparrow)\) 
& \textbf{\(\bar{L}\)}\((\downarrow)\) 
& \textbf{MSE}\((\downarrow)\) 
& \textbf{\(\rho\)}\((\uparrow)\) 
& \textbf{\(\bar{R}\)}\((\uparrow)\) 
& \textbf{\(\bar{L}\)}\((\downarrow)\)  \\ \hline

ATIL 
& \textbf{4.15E-05} & \textbf{100} & \textbf{1.000} & \textbf{0.054} 
& 1.837E-04 & \textbf{100} & 0.992 & 0.166
& \textbf{4.10E-05} & \textbf{100} & \textbf{1.000} & \textbf{0.0917}
& 1.84E-04 & \textbf{100} & 0.9920 & 0.1884 \\ %

FICNNIL 
& 9.50E-05 & 94 & 0.995 & 0.079 
& 4.150E-05 & 93 & 0.998 & 0.083
& 6.50E-05 & 99 & 0.9959 & 0.1071
& 2.80E-05 & 99 & 0.9997 & 0.0746 \\ %

PICNNIL 
& 6.91E-05 & 94 & 0.997 & 0.068 
& \textbf{2.597E-05} & 93 & \textbf{1.000} & \textbf{0.063}
& 5.40E-05 & 99 & 0.9966 & 0.0992
& \textbf{2.50E-05} & \textbf{100} & \textbf{1.0000} & \textbf{0.0672} \\ \hline

FICNN-PEARL\cite{grzeskiewicz2025uncovering_ICNN_Utility_Max} 
& 3.03E-02 & 76 & 0.8670 & 0.8302
& 1.32E-02 & 73 & 0.9273 & 0.5891
& 1.07E-02 & 90 & 0.9309 & 0.5309
& 1.17E-02 & 80 & 0.9313 & 0.3424 \\ %

PICNN-PEARL\cite{grzeskiewicz2025uncovering_ICNN_Utility_Max} 
& 3.14E-02 & 56 & 0.8684 & 0.8266
& 4.12E-02 & 56 & 0.8600 & 0.8506
& 4.64E-02 & 56 & 0.8265 & 0.8963
& 9.78E-02 & 61 & 0.7046 & 0.9195 \\

LinearRegression 
& 5.78E-04 & 86 & 0.955 & 0.058 
& 1.827E-04 & 81 & 0.983 & 0.105
& 1.76E-04 & 99 & 0.9840 & 0.1813
& 1.47E-04 & 97 & 0.9906 & 0.2325 \\ %

Ridge 
& 2.87E-04 & 86 & 0.976 & 0.082 
& 1.984E-04 & 79 & 0.988 & 0.149
& 1.75E-04 & 99 & 0.9844 & 0.1833
& 1.48E-04 & 97 & 0.9909 & 0.2352 \\ %

Lasso 
& 2.48E-03 & 9 & 0.978 & 0.619 
& 2.327E-03 & 16 & 0.966 & 0.869
& 2.216E-03 & 19 & 0.9654 & 0.8739
& 2.317E-03 & 24 & 0.9658 & 1.0000 \\ %

DecisionTree 
& 1.78E-03 & 42 & 0.980 & 0.399 
& 9.905E-04 & 40 & 0.972 & 0.449
& 2.10E-04 & 72 & 0.9938 & 0.2368
& 3.02E-04 & 58 & 0.9878 & 0.2380 \\ %

RandomForest 
& 8.50E-04 & 46 & 0.994 & 0.359 
& 6.950E-04 & 61 & 0.988 & 0.464
& 7.50E-05 & 85 & 0.9988 & 0.1646
& 1.06E-04 & 79 & 0.9986 & 0.2007 \\ %

MLPRegressor 
& 9.40E-03 & 35 & 0.962 & 1.000 
& 4.889E-03 & 44 & 0.943 & 1.000
& 4.727E-03 & 52 & 0.9499 & 1.0000
& 2.922E-03 & 60 & 0.9382 & 0.7545 \\ \hline
\end{tabular}%
}
\textbf{Note}: \(\rho\): Ordinal Accuracy(\(\%\));\;\; \(\bar{R}\): Max Normalized Average Throughput;\;\;\(\bar{L}\): Max Normalized Packet Loss. Number of test samples \(=100\)
\end{table*}

\paragraph*{Simulation Parameters}

The RSU is assumed to be located at the origin, and the initial positions of all three vehicles in each simulation iteration are randomly generated within the RSU footprint. Further, the RSU employs a Kalman filter, under a constant-velocity target dynamics model, to track the vehicles. The RSU utilizes the ISAC-AB strategy in~\eqref{eq:ISAC_AB_Objective} to adaptively allocate the beam splitting ratio, and hence the communication data rate, based on the tracking accuracy
\[
\alpha_{t,m}=\operatorname{trace}\!\left(\boldsymbol{\Sigma}_{t\mid t-1,m}^{-1}\right).
\]
The RSU utility function used for simulation is given by~\eqref{eq:ISACABUtil}, with \(R(W)=1\) and \(R(N)=4R(W)\). The parameters \(w_m\), \(\delta_{t,m}\), and \(\sigma_{\phi}\) are set to unity.

{\sl Fixed and State Dependent Utility Functions:} The dependence of the RSU utility on the state variable \(\boldsymbol{\alpha}_t\) is captured through the weighting function \(p(\boldsymbol{\alpha}_t):\boldsymbol{\alpha}_t\mapsto[0,1]\) in~\eqref{eq:ISACABUtil}, which reflects the relative preference between sensing-oriented and communication-oriented operation. Such weighting is motivated by the utility based wireless resource allocation\cite{kuo2007utility_P_alpha} and ISAC literature~\cite{du2023v2i_isac_P_alpha,Liu2023ISAC_VN_Chapter,al2024ISACresources}, where sensing--communication tradeoffs are commonly modeled through weighted objective functions or weighted sum-rate formulations. In the fixed-utility setting, \(p(\boldsymbol{\alpha}_t)\) is chosen as a constant \(p_0 = 0.5\), signifying equal weightage for widebeam (sensing) and narrowbeam (communication). 
For numerical simulations in the state-dependent utility setting, following~\cite{kuo2007utility_P_alpha}, we model \(p(\boldsymbol{\alpha}_t)\) as a mapping from the average tracking accuracy, \(\frac{1}{M}\sum_{m=1}^{M}\alpha_{t,m}\), to the continuous set \([0,1]\). However, to avoid the extreme values \(0\) and \(1\), corresponding to only wide-beam or only narrow-beam allocation, respectively, we model \(p(\boldsymbol{\alpha}_t)\in [0.1,0.9]\) as follows:
\begin{equation}\label{eq:p_alpha_num}
p(\boldsymbol{\alpha}_t)=
\begin{cases}
p_0 = 0.5, & \text{fixed utility},\\[1ex]
0.1+0.8\,\bar{z}(\boldsymbol{\alpha}_t), & \text{state-dependent utility},
\end{cases}
\end{equation}
where
\begin{equation}\label{eq:zbar_alpha_num}
\bar{z}(\boldsymbol{\alpha}_t)
=
\mathrm{clip}\!\left(
\frac{\frac{1}{M}\sum_{m=1}^{M}\alpha_{t,m}-\alpha_{\min}}
{\alpha_{\max}-\alpha_{\min}},
\,0,\,1
\right),
\end{equation}
with \(\mathrm{clip}(x,0,1)=\min\{\max\{x,0\},1\}\). In \eqref{eq:zbar_alpha_num}, \(\alpha_{t,m}\) denotes the ISAC RSU's tracking accuracy for vehicle \(m\). Further, \(\alpha_{\min}=1.5\) and \(\alpha_{\max}=8\) are chosen in \eqref{eq:zbar_alpha_num} as design constants consistent with the approximate minimum and maximum values of the individual-vehicle tracking accuracies observed in the simulation.

For each experiment, under the centralized learning paradigm, \(T\) input-output observations of the form
\[
\mathcal{D}_{1:T}=\{(\boldsymbol{\alpha}_t,\boldsymbol{\beta}_t)\}_{t=1}^{T}
\]
are generated and assumed to be available to the VMC for inverse learning. The numerical results for the centralized learning setting, for both static and dynamic utility models and using ATIL, FICNNIL, and PICNNIL, are presented in Sections~\ref{subsec:Central_Reconstr_Util}--\ref{subsec:Central_Cluster_Head}. The numerical results corresponding to the federated inverse learning paradigm are presented later in Sec.~\ref{subsec:Federated_Num_Res}.

\subsection{Inverse learning -- Reconstruction of the utility function}\label{subsec:Central_Reconstr_Util}

The effectiveness of the proposed inverse learning approaches, namely ATIL, FICNNIL, and PICNNIL, in reconstructing the RSU decision-making behavior is evaluated for both static and dynamic utility settings. The performance is quantified in terms of the accuracy of capturing the ordinal preferences of beam allocation across the vehicles, defined as
\begin{equation}\label{eq:num_ord_pref}
    \rho
    =
    \frac{1}{M T_{test}}
    \sum_{t'=1}^{T_{test}}
    \sum_{m=1}^{M}
    \mathbf{1}_{\{(rank(\beta_{t',m})-rank(\beta_{t',m}^{(\mathrm{pred})}))=0\}},
\end{equation}
and the mean squared error
\begin{equation}\label{eq:mse_num}
\mathrm{MSE}
=
\frac{1}{T_{test}}
\sum_{t'=1}^{T_{test}}
\left\|
\boldsymbol{\beta}_{t'}-\boldsymbol{\beta}_{t'}^{(\mathrm{pred})}
\right\|^2,
\end{equation}
between the predicted and actual RSU actions over the test probes.

The ``MSE'' and ``Ordinal Accuracy'' columns of Table~\ref{tab:il_comparison_t20_t200} report the IL performance in a centralized setting for both static and dynamic utility formulations. The results are presented for different numbers of training samples, namely $T=20$ and $T=200$, while the number of test samples is fixed at $100$ across all cases. 
Table~\ref{tab:il_comparison_t20_t200} further provides a comparative evaluation of ATIL, FICNNIL, and PICNNIL against standard supervised regression models from the classical machine learning literature, which aim to learn the mapping from the state of the world $\boldsymbol{\alpha}_t$ to the corresponding RSU action $\boldsymbol{\beta}_t$.

{\sl Effectiveness in capturing ordinal preferences:} 
Correct identification of ordinal preferences implies that the predicted beam-splitting ratios preserve the ordering of RSU–vehicle data rates from highest to lowest. The accuracy of ordinal preference recovery, quantified by $\rho$ in~\eqref{eq:num_ord_pref}, is reported in Table~\ref{tab:il_comparison_t20_t200}. 

Due to the structure of the ISAC-AB RSU utility in~\eqref{eq:ISACABUtil}, along with its dependence on the environmental state as defined in~\eqref{eq:p_alpha_num}, the ordinal preferences over RSU actions remain identical for both fixed and state dependent utility formulations for any given $\boldsymbol{\alpha}_t$. Consequently, ATIL, which directly leverages Afriat's Theorem (Theorem~\ref{Th:Afriats}) for nonparametric utility reconstruction, achieves perfect recovery of ordinal preferences ($\rho = 100\%$) across both utility settings and for all training sample sizes.

In contrast, FICNNIL and PICNNIL, which rely on a revealed preference-inspired bilevel optimization framework with structured parametric models, achieve an ordinal accuracy of approximately $94\%$ for $T=20$ training samples and nearly $100\%$ for $T=200$ in the static utility setting. In the dynamic setting, both methods attain an ordinal accuracy of $93\%$ for $T=20$, while for $T=200$, FICNNIL and PICNNIL achieve $99\%$ and $100\%$, respectively. The marginally improved performance of PICNNIL over FICNNIL in the dynamic utility setting can be attributed to the ability of the PICNN architecture to explicitly capture the dependence of the utility function on the environmental state $\boldsymbol{\alpha}_t$.

These results indicate that, with sufficient training data, the ordinal preference recovery performance of the proposed parametric IL frameworks (FICNNIL and PICNNIL) approaches the theoretical optimum achieved by ATIL. The marginal degradation observed for smaller training sample sizes can be attributed to the inherent data requirements of parametric machine learning models. Nevertheless, as evident from Table~\ref{tab:il_comparison_t20_t200}, irrespective of the number of training samples, the ordinal accuracy achieved by ATIL, FICNNIL, and PICNNIL significantly outperforms that of classical supervised learning benchmarks, which aim to directly learn a mapping from $\boldsymbol{\alpha}_t$ to $\boldsymbol{\beta}_t$ without exploiting the underlying preference structure.
Conventional data-driven learning models typically pose IL as a supervised regression/classification task \cite{gweon2023BehClone} that maps the input $\boldsymbol{\alpha}$ (state of the world-kinematic states of vehicles) to the output $\boldsymbol{\beta}$ (RSU action) by minimizing empirical risk, e.g., mean-squared error (MSE) for regression or categorical cross-entropy for classification, over observed input--output pairs $(\boldsymbol{\alpha},\boldsymbol{\beta})$.  
Since the RSU's rational behavior is assumed to be consistent with \eqref{eq:RSU_Strategy_general}, IL models and algorithms, that we derive in this paper,
explicitly exploit this structural knowledge.

{\sl MSE performance:} 
Although the ordinal preferences of the RSU remain identical under both fixed and state dependent utility formulations, the corresponding optimal beam-splitting ratios—and hence the achievable data rates—differ significantly across these settings. The MSE measures the discrepancy between the predicted RSU actions obtained via inverse learning and the true RSU actions, and is therefore a critical metric influencing throughput and packet error performance in predictive scheduling. The \textit{MSE} columns in Table~\ref{tab:il_comparison_t20_t200} provide a comparative evaluation of different IL approaches under both fixed and state dependent utility settings in a centralized learning framework.

In the fixed utility setting, where the RSU utility is invariant with respect to the environmental state $\boldsymbol{\alpha}_t$, ATIL achieves the lowest MSE across all training sample sizes, consistent with the guarantees provided by Afriat's Theorem. FICNNIL and PICNNIL also achieve comparably low MSE values on the order of $10^{-5}$. In contrast, classical supervised regression benchmarks exhibit significantly higher MSE, at least an order of magnitude larger (i.e., $10^{-4}$ or higher), irrespective of the number of training samples.

In the state dependent utility setting, FICNNIL and PICNNIL continue to maintain low MSE on the order of $10^{-5}$. Notably, PICNNIL achieves the best MSE of $2.6 \times 10^{-5}$ for $T=20$, which is approximately $50\%$ lower than the next best value of $4.2 \times 10^{-5}$ achieved by FICNNIL. In contrast, ATIL—due to its assumption that the utility depends solely on $\boldsymbol{\beta}_t$—attains an MSE of $1.84 \times 10^{-4}$, which is nearly an order of magnitude higher than its performance in the fixed utility setting. Nevertheless, the MSE achieved by FICNNIL and PICNNIL remains consistently lower than that of classical supervised learning benchmarks.

\subsection{IL-Assisted Predictive Scheduling for Cooperative Data Downloading in VMC}

We employ Algorithm~\ref{ATIL_Based_Pred_Scheduling} to schedule the maximum possible data download for each vehicle within a unit-time epoch. Let \(\beta_{t,m}(i)\) and \(R_{t,m}(i)\) denote, respectively, the predicted RSU beam allocation and the corresponding predicted downlink data rate obtained using inverse learning model \(i\), where \(i\in\{\text{ATIL},\text{FICNNIL},\text{PICNNIL}\}\) or one of the baseline regression models. The beam-splitting ratios obtained using the respective IL/regression models are used to perform predictive scheduling, as given in Steps~\(4\) through~\(6\) of Algorithm~\ref{ATIL_Based_Pred_Scheduling}.

If the scheduled data size \(R_{t,m}(i)\) exceeds the actual available rate \(R_{t,m}\), the excess amount is treated as lost data. On the other hand, if \(R_{t,m}(i)<R_{t,m}\), the remaining capacity is left unused, but no packet loss is incurred. The predictive scheduling performance is evaluated in terms of the maximum normalized throughput
\begin{equation}\label{eq:maxNormThrou}
    R_t'(i)=\frac{\sum_m R_{t,m}(i)}{\max_j \sum_m R_{t,m}(j)},
\end{equation}
and the maximum normalized packet loss
\begin{equation}\label{eq:packetLoss}
    \hspace{-1em}L'(i)=\frac{\sum_m \big(R_{t,m}(i)-R_{t,m}\big)\mathbf{1}_{\{R_{t,m}(i)-R_{t,m}>0\}}}{\max_j \sum_m \big(R_{t,m}(j)-R_{t,m}\big)\mathbf{1}_{\{R_{t,m}(j)-R_{t,m}>0\}}}.
\end{equation}

The \textit{Avg Throughput (\(\bar{R}\))} and \textit{Packet Loss (\(\bar{L}\))} columns in Table~\ref{tab:il_comparison_t20_t200} (both normalized by their respective column-wise maximum values for ease of comparison) report the predictive scheduling performance under a centralized learning setting. The throughput and packet loss metrics are directly influenced by the MSE performance (see Sec.~\ref{subsec:Central_Reconstr_Util}) of the corresponding models.
In the fixed utility setting, ATIL achieves the highest average throughput and the lowest packet loss. In contrast, under the state dependent utility setting, PICNNIL demonstrates the best performance in terms of both throughput and packet loss, followed by FICNNIL and ATIL, respectively. Overall, the proposed IL approaches—ATIL, FICNNIL, and PICNNIL—consistently outperform classical supervised learning benchmarks in predictive scheduling performance.
For instance, in the case of $T=200$ training samples, the packet loss achieved by the best-performing IL-based approach is, on average, approximately $50\%$ and $75\%$ lower than that of the best-performing classical machine learning benchmarks in the fixed and state dependent settings, respectively. This reduction indicates a significant decrease in packet retransmissions and energy expenditure, along with improved latency performance.

\subsection{IL-assisted dynamic cluster head selection}\label{subsec:Central_Cluster_Head}

We evaluate VMC cluster head selection under the following strategies:
\begin{inparaenum}[(i)]
    \item the proposed inverse learning-assisted dynamic link selection, using ATIL, FICNNIL, or PICNNIL predictions, as discussed in Sec.~\ref{sec:ILDynamic},
    \item assigning the first vehicle in the VMC to establish communication (\texttt{First\_Fixed})~\cite{zhao2025platoon01,Platooning2023proximity1}, and
    \item assigning the vehicle closest to either of the RSUs to establish communication (\texttt{Nearest\_Fixed})~\cite{zhao2025platoon01,Platooning2023proximity1}.
\end{inparaenum}

Table~\ref{tab:leader_comparison_combined} reports the normalized throughput achieved by the selected cluster head using different cluster head selection strategies for varying test dataset sizes under both fixed and state dependent utility settings, with the number of training samples fixed at $T=200$. In each case, the throughput values are normalized with respect to the maximum throughput attained among all compared strategies.
The results indicate that the proposed IL-assisted dynamic selection methods consistently identify the vehicle with the most favorable communication link to the RSU as the cluster head, thereby outperforming the fixed selection baselines.
Overall, these findings further demonstrate the effectiveness of the proposed IL-based framework for adaptive cluster head selection in 6G-ISAC vehicular networks.

\begin{table}[t]
\centering
\caption{Comparison of cluster head selection methods in VMC}
\label{tab:leader_comparison_combined}
\resizebox{\columnwidth}{!}{%
\begin{tabular}{|c|ccccc|}
\hline
\multicolumn{6}{|c|}{\textbf{Fixed Utility}} \\ \hline
\textbf{Test Data Size} 
& \textbf{ATIL} 
& \textbf{FICNNIL} 
& \textbf{PICNNIL} 
& \textbf{First\_Fixed} 
& \textbf{Nearest\_Fixed} \\ \hline

20  
& \textbf{1.0} & \textbf{0.95} & \textbf{0.95} & 0.84 & 0.66 \\

200 
& \textbf{1.0} & \textbf{0.95} & \textbf{1.0} & 0.85 & 0.67 \\ \hline

\multicolumn{6}{|c|}{\textbf{State Dependent Utility}} \\ \hline
\textbf{Test Data Size} 
& \textbf{ATIL} 
& \textbf{FICNNIL} 
& \textbf{PICNNIL} 
& \textbf{First\_Fixed} 
& \textbf{Nearest\_Fixed} \\ \hline

20  
& \textbf{1.0} & \textbf{0.90} & \textbf{1.0} & 0.82 & 0.64 \\

200 
& \textbf{1.0} & \textbf{0.95} & \textbf{1.0} & 0.85 & 0.67 \\ \hline

\end{tabular}%
}
\end{table}

\subsection{Federated inverse learning -- Reconstruction of the utility function}\label{subsec:Federated_Num_Res}
\begin{figure}[h]
    \centering
    \includegraphics[width=0.8\linewidth]{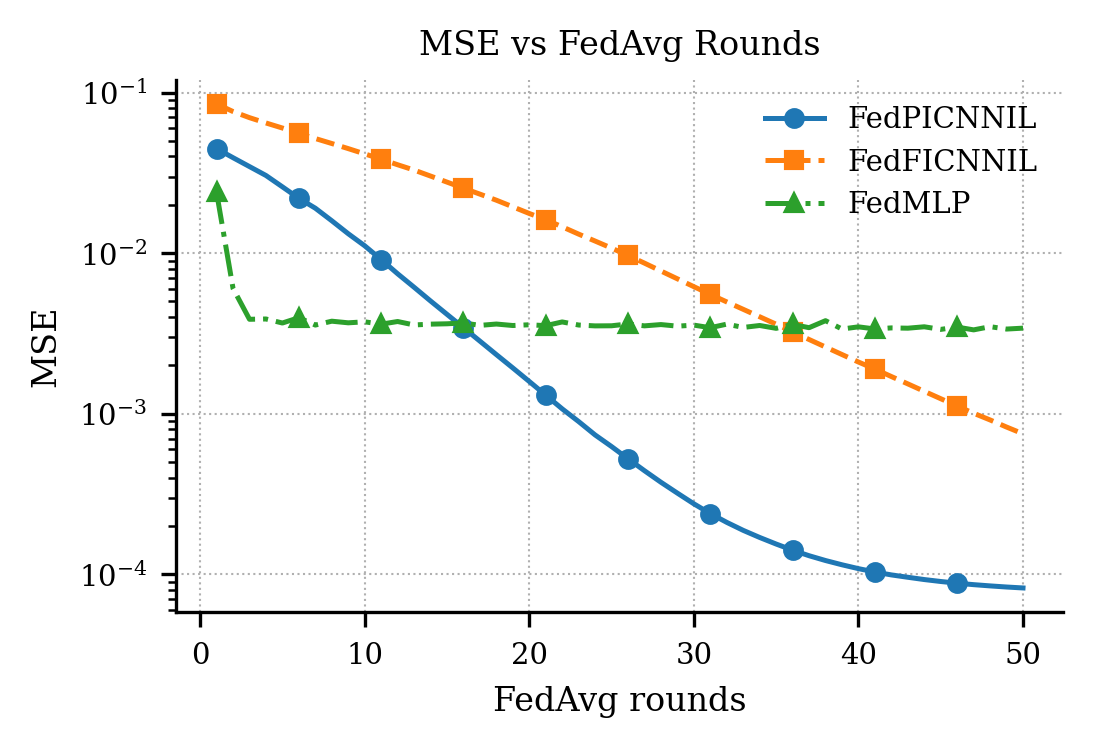}
\caption{Prediction MSE versus federated averaging rounds. The proposed ICNN-based federated inverse learning method steadily reduces MSE, while the naive MLP-based baseline saturates after a few initial rounds.}
    \label{fig:mse_vs_rounds_fed}
\end{figure}

\begin{figure}[t]
    \centering
    \includegraphics[width=0.8\linewidth]{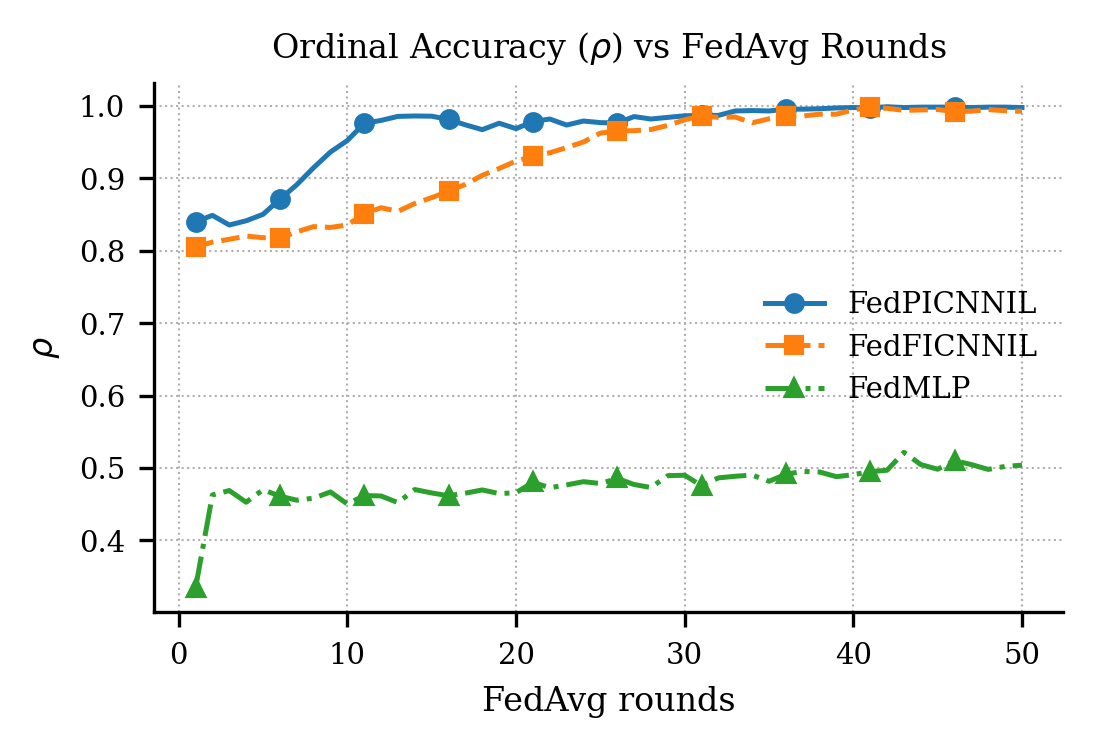}
    \caption{Prediction ordinal accuracy vs Federated averaging round. The proposed ICNN-based federated inverse learning method steadily improves and achieves \(\sim 100\%\) ordinal accuracy, while the naive MLP-based baseline saturates around \(\sim 50\%\) after a few initial rounds.}
    \label{fig:rho_vs_rounds_fed}
\end{figure}
We next evaluate the proposed federated inverse learning approaches, namely FedFICNNIL and FedPICNNIL, in Algorithm~\ref{alg:Federated_ICNNIL}. In this setting, the revealed preference dataset is distributed across multiple VMC clients, and only model parameters are communicated to the server for aggregation, without sharing raw local observations. We focus on the state dependent utility setting, wherein significant heterogeneity exists across local datasets available to individual VMCs. 

The federated IL experiments are carried out using \(K=2\) VMC clients, denoted by Client 1 and Client 2, as shown in Fig.~\ref{fig:BT_ATIL_Fed}. To emulate statistical heterogeneity across clients in the dynamic utility setting (equations \eqref{eq:ISACABUtil} and \eqref{eq:p_alpha_num}), the local datasets are generated such that the proportions of samples corresponding to low values of \(p(\boldsymbol{\alpha})\), i.e., \(p(\boldsymbol{\alpha})<0.5\), are \(70\%\) and \(30\%\) for Clients 1 and 2, respectively. Thus, the two clients observe different local operating regimes of the ISAC RSU, resulting in a non-IID revealed preference distribution across the federated network.
The federated inverse learning performance is evaluated using the same metrics as in the centralized setting, namely the ordinal preference accuracy \(\rho\) in~\eqref{eq:maxNormThrou} and the MSE between the predicted and actual RSU actions in \eqref{eq:mse_num}. The objective here is to assess whether the global model learned through federated averaging can accurately reconstruct the RSU decision-making behavior despite the presence of heterogeneous local data distributions.
\begin{table}[t]\centering
\caption{Comparison of Federated Inverse Learning Performance (Dynamic Utility)}
\label{tab:IL_compare_Dynamic_Fed_transposed}
\resizebox{0.4\textwidth}{!}{%
\scriptsize
\begin{tabular}{|l|cc|cc|}
\hline
\multirow{2}{*}{\textbf{IL Model}} 
& \multicolumn{2}{c|}{\textbf{$T=20$}} 
& \multicolumn{2}{c|}{\textbf{$T=200$}} \\ \cline{2-5}
& \textbf{$\rho$ (\%)} & \textbf{MSE} 
& \textbf{$\rho$ (\%)} & \textbf{MSE} \\ \hline

\textbf{FedFICNNIL} 
& \textbf{98.6} & \textbf{9.5E-4} 
& \textbf{99.2} & \textbf{7.3E-4} \\ \hline

\textbf{FedPICNNIL} 
& \textbf{98.6} & \textbf{9.7E-5} 
& \textbf{100} & \textbf{7.6E-5} \\ \hline

\textbf{FedMLP} 
& 3.4 & 8.6E-1 
& 50 & 4.1E-3 \\ \hline

\end{tabular}
}
\end{table}

From Figures~\ref{fig:mse_vs_rounds_fed} and~\ref{fig:rho_vs_rounds_fed}, along with Table~\ref{tab:IL_compare_Dynamic_Fed_transposed}, it is observed that FedPICNNIL achieves the highest ordinal (rank) accuracy and the lowest MSE in the federated setting under data heterogeneity induced by the state dependent nature of the ISAC RSU utility. This behavior is expected, as the local parameter updates in FedPICNNIL explicitly account for the dependence of the utility function on the environmental state $\boldsymbol{\alpha}_t$, whereas FedFICNNIL does not.
Consequently, as illustrated in Figures~\ref{fig:mse_vs_rounds_fed} and~\ref{fig:rho_vs_rounds_fed}, FedPICNNIL attains nearly $100\%$ ordinal accuracy and an MSE on the order of $10^{-5}$ within $50$ training rounds. In contrast, FedFICNNIL achieves a slightly lower ordinal accuracy of approximately $97\%$ and an MSE on the order of $10^{-4}$. 
As a baseline, we compare FedPICNNIL and FedFICNNIL with a FedAvg extension of a supervised MLP regressor that maps \(\boldsymbol{\alpha}_t\) to the corresponding \(\boldsymbol{\beta}_t\) by minimizing the mean-squared error. Both FedPICNNIL and FedFICNNIL significantly outperform the classical MLP regression baseline in the federated setting, as evidenced by the results in Table~\ref{tab:IL_compare_Dynamic_Fed_transposed}, in terms of both ordinal accuracy and MSE.

\section{Conclusions}\label{sec:Conclusion}

This paper presented data-centric inverse learning (IL) frameworks—namely ATIL, FICNNIL, PICNNIL, FedFICNNIL and FedPICNNIL—for characterizing the strategic behavior of an autonomous intent-driven 6G-ISAC RSU. 
We proposed a novel bilevel training framework for training FICNN and PICNN networks.
The proposed IL frameworks leverage revealed preference theory from microeconomics to infer the underlying RSU utility function from observed actions.
The effectiveness of the IL-based adaptive communication strategies was demonstrated for vehicular networks through two applications:
\begin{inparaenum}[(i)]
    \item predictive scheduling for cooperative data downloading, and
    \item channel-capacity-aware predictive cluster head selection,
\end{inparaenum}
under both fixed and state dependent RSU utility settings, and within both centralized and federated learning paradigms.
ATIL is particularly suitable for scenarios involving fixed RSU utilities under centralized learning settings. In contrast, the proposed parametric approaches exhibit greater generalizability and improved performance in state dependent and federated settings. In particular, FedPICNNIL is well-suited for inverse learning in federated environments with significant statistical heterogeneity across local VMC datasets, owing to its ability to model state-dependent utilities.
Empirical results demonstrate that the proposed IL-based strategies consistently outperform conventional non-learning-based approaches as well as classical machine learning methods in terms of throughput and packet loss, thereby highlighting their effectiveness for adaptive V2I communication in 6G-ISAC vehicular networks.

\bibliographystyle{IEEEtran}
\bibliography{ref_short}
\end{document}